\documentclass[a4paper]{article}

\def\final{0}  
\def\iflong{\iffalse}
\ifnum\final=0  
\newcommand{\knote}[1]{[{\tiny Karthik: \bf #1}]\marginpar{*}}
\newcommand{\jnote}[1]{[{\tiny Jochen: \bf #1}]\marginpar{*}}
\else 
\newcommand{\knote}[1]{}
\newcommand{\jnote}[1]{}
\fi  

\usepackage[T1]{fontenc}
\usepackage{units}
\usepackage{graphicx}
\usepackage{comment,todonotes}
\usepackage{xcolor}
\usepackage{tikz}
\usepackage{enumerate}
\usetikzlibrary{decorations.pathreplacing}
\usetikzlibrary{shadows,arrows}
\usepackage{listings}
\usepackage{amsthm}
\usepackage{mathtools}
\DeclarePairedDelimiter{\ceil}{\lceil}{\rceil}

\usepackage{thm-restate}
\usepackage{bbm}
\usepackage{mathbbol}
\usepackage{authblk}

\usepackage{subcaption}
\theoremstyle{plain}
\newtheorem{theorem}{Theorem}
\newtheorem{lemma}[theorem]{Lemma}
\newtheorem{proposition}[theorem]{Proposition}

\theoremstyle{remark}
\newtheorem{claim}[theorem]{Claim}
\newtheorem{remark}[theorem]{Remark}

\theoremstyle{definition}
\newtheorem{definition}{Definition}

\renewcommand{\mid}{:}
\DeclareMathOperator*{\argmax}{\arg\!\max}
\newcommand{\R}{\mathbb{R}}

\newcommand{\0}{\mathbb{0}}
\newcommand{\1}{\mathbb{1}}
\newcommand{\st}{\ensuremath{\mbox{s.t.}}}

\def\FULL{conf} 
\newcommand{\iffull}[1]{\ifthenelse{\equal {\FULL}{full}}{#1}{}}
\newcommand{\ifconf}[1]{\ifthenelse{\equal {\FULL}{full}}{}{#1}}

\title{Additive Stabilizers for Unstable Graphs}

\author[1]{Karthekeyan Chandrasekaran}
\author[2]{Corinna Gottschalk}
\author[3]{Jochen K\"onemann}
\author[2]{Britta Peis}
\author[2]{Daniel Schmand}
\author[2]{Andreas Wierz}

\affil[1]{University of Illinois Urbana-Champaign, Department of
  Industrial and Enterprise Systems Engineering,
  \texttt{karthe@illinois.edu}}
\affil[2]{RWTH Aachen University, School of Business and Economics,
  \texttt{\{gottschalk,peis,schmand,wierz\}@oms.rwth-aachen.de}}
\affil[3]{University of Waterloo, Department of Combinatorics \&
  Optimization, \texttt{jochen@uwaterloo.ca}}

\begin{document}
\maketitle

\begin{abstract}
	Stabilization of graphs has received substantial attention in recent years due to its connection to game theory. Stable graphs are exactly
  the graphs inducing a matching game with non-empty core. They are also the graphs that induce a network bargaining game with a balanced solution. A graph with
  weighted edges is called stable if the maximum weight of an integral
  matching equals the cost of a minimum fractional weighted vertex
  cover. If a graph is not stable, it can be stabilized in different
  ways. Recent papers have considered the deletion or addition of
  edges and vertices in order to stabilize a graph. In this
  work, we focus on a fine-grained stabilization strategy, namely stabilization of graphs by fractionally increasing edge weights. 
  
  We show the following results for stabilization by minimum weight increase in edge weights (min additive stabilizer): (i) Any approximation algorithm for min additive stabilizer that achieves a factor of $O(|V|^{1/24-\epsilon})$ for $\epsilon>0$ would lead to improvements in the approximability of densest-$k$-subgraph. (ii) Min additive stabilizer has no $o(\log{|V|})$ approximation unless NP=P. Results (i) and (ii) together provide the first super-constant hardness results for any graph stabilization problem. On the algorithmic side, we present (iii) an algorithm to solve min additive stabilizer in factor-critical graphs exactly in poly-time, (iv) an algorithm to solve min additive stabilizer in arbitrary-graphs exactly in time exponential in the size of the Tutte set, and (v) a poly-time algorithm with approximation factor at most $\sqrt{|V|}$ for a super-class of the instances generated in our hardness proofs.

\end{abstract}

\section{Introduction}\label{sec:intro}

Over the last two decades, algorithmic game theory has established
itself as a vibrant and rich subarea of theoretical computer science
as is evidenced by several recent books (e.g., see
\cite{CEW11,Nisan:2007,SL08}). A crucial driver in this development
is the increasingly networked structure of today's society, and the
impact this development has on the day-to-day interactions that humans
engage in. Finding and analyzing graph-theoretic models for such
networks is at the heart of the field of network exchange theory, and
is captured by the two recent books \cite{EK10,Ja08}. 

Social networks, and the interaction of individuals in those also
motivated our work. Specifically, our interest started with
\cite{KT08}, where Kleinberg
\& Tardos introduce {\em network bargaining} as a
natural extension of Nash's classical two-player bargaining
game~\cite{Na50} to the network setting. The players in
Kleinberg and Tardos' game correspond to the vertices in an underlying
graph $G=(V,E)$. Each $\{u,v\} \in E$ corresponds to a potential deal
of given value $w_{uv} \geq 0$. Each player is allowed to interact with the neighbors to agree upon a sharing of the value on the edge between them and eventually arrive at a deal with at most one of her neighbours. 
Therefore, outcomes in network bargaining
correspond to {\em matchings} $M \subseteq E$, and
an {\em allocation} $y \in \R^V_+$ of $w(M)$ to the players.  In
particular, we want $y_u+y_v=w_{uv}$ for all $\{u,v\} \in M$, and
$y_u=0$ if $u$ is not incident to an edge of $M$ ($u$ is {\em
  exposed}). 

Kleinberg and Tardos introduce the concept of {\em stability}, and
call an allocation $y$ to be {\em stable} if $y_u+y_v \geq w_{uv}$ for {\em
  all} edges $\{u,v\} \in E$. Naturally extending Nash's bargaining
solution, the authors define the {\em outside option} $\alpha_u$ of
a player $u$ given an allocation $y$ as the largest value that $u$ can {\em
  extract} from one of its neighbours. An allocation $y$ is then deemed to be
{\em balanced} if the value of each matching edge $\{u,v\} \in M$ is
split according to Nash's bargaining condition: each player
$a \in \{u,v\}$ receives its outside option $\alpha_a$, and the
  remaining value of $\{u,v\}$ is divided equally among the players. 
One of Kleinberg and Tardos' main results is that balanced outcomes
exist in a given network bargaining instance if and only if stable
ones exist, and these can be computed efficiently.

Network bargaining is closely related to the 
cooperative {\em matching game} introduced by Shapley and
Shubik~\cite{shapley1971assignment}, where the player set once more
corresponds to the vertices of an underlying graph $G=(V,E)$, and the
characteristic function assigns the maximum weight of a matching in
$G[S]$ to each set $S \subseteq V$ of vertices. The {\em core} of an
instance of this game consists of allocations $y \in \R^V_+$ of the
weight $\nu(G,w)$ of a maximum-weight matching to the players such
that $y_u + y_v \geq w_{uv}$ for all $\{u,v\} \in E$. Hence, core
allocations exactly correspond to stable allocations in network
bargaining (this observation was recently also made by Bateni et
al.~\cite{BH+10}). 

A given instance of network bargaining therefore has a stable (and
also a balanced) outcome if and only if the core of the corresponding
matching game is non-empty. We state the classical maximum weight matching LP that
has a variable $x_e$ for each edge $e \in E$ (we use
$x(\delta(v))$ as a convenient short-hand for $\sum_{e \in
  \delta(v)}x_e$, where $\delta(v)$ denotes the set of edges
incident to $v$):
\begin{equation}\tag{P}\label{lp:matching}
  \nu_f(G,w) := \max \{ \sum_{e \in E} w_ex_e \,:\, x(\delta(v)) \leq
  1 ~ \mbox{ for all } v \in V, x \geq
  \0. \}. 
\end{equation}
The linear programming dual of \eqref{lp:matching} has a variable
$y_v$ for each vertex $v \in V$, and a covering constraint for each edge $e \in
E$:
\begin{equation}\tag{D}\label{lp:cover}
  \tau_f(G,w) := \min \{ \sum_{v \in V} y_v \,:\, y_u + y_v \geq
  w_{uv} \mbox{ for all } \{u,v\} \in E, y \geq \0 \}.
\end{equation}
Feasible solutions of \eqref{lp:matching} and \eqref{lp:cover} will
henceforth be referred to as {\em fractional matchings} and {\em
  fractional $w$-vertex covers}, respectively. In the unit-weight special case,
where $w=\1$, we will omit the argument $w$ from the $\nu$ and $\tau$
notation for brevity. 
An immediate observation is that a given
instance of network bargaining has a stable outcome iff the core of the
corresponding matching game is non-empty iff 
$\nu(G,w) = \nu_f(G,w) = \tau_f(G,w)$,
where the second equality follows from linear programming duality. 
In other words, stable outcomes exist iff LP \eqref{lp:matching} admits
integral optimum solutions. We will call a (possibly weighted) graph
{\em stable} if the induced network bargaining instance admits a
stable outcome. 

\iffalse
Given an undirected graph $G=(V,E)$, a \emph{matching} is a subset of
edges $M\subseteq E$ such that every vertex is adjacent to at most one
edge in $M$. A \emph{vertex cover} is a subset of vertices
$C\subseteq V$ such that every edge $\{u,v\}\in E$ has at least one
end vertex in $C$. The problems of computing the maximum cardinality
of a matching, denoted by $\nu(G)$, and the minimum cardinality of a
vertex cover, denoted by $\tau(G)$, are well-studied in the literature
of combinatorial optimization \cite{schrijver2003combinatorial}. It is
well-known that the LP-relaxations
$(P)\ \nu_f(G):=\max_{x\ge 0} \{\sum_{e\in E} x_e \mid \sum_{e\in
  \delta(v)} x_e \le 1,\ \forall v\in V\}$
and
$(D)\ \tau_f(G):=\min_{y\ge 0}\{\sum_{v\in V} y_v \mid y_u+y_v\ge 1 \
\forall \{u,v\}\in E\}$
of the maximum cardinality matching and the minimum cardinality vertex
cover problem, respectively, are dual to each other. As usual, we call
a feasible solution $x$ of $(P)$ a \emph{fractional matching}, and a
feasible solution $y$ of $(D)$ a \emph{fractional vertex cover}.
Strong duality of linear programs implies
$\nu(G)\le \nu_f(G)= \tau_f(G)\le \tau(G).$
While $\nu(G)$ can be computed in polynomial time (e.g., see
\cite{edmonds1965paths}), computing $\tau(G)$ is
NP-hard \cite{karp1972reducibility}. Assuming the unique
games conjecture, $\tau(G)$ cannot be approximated within any
multiplicative factor below two \cite{khot2008vertex}. Still, optimal
\emph{half-integral} solutions of (P) and (D) exist and can be
computed efficiently \cite{balinski1965integer}.

\ifconf{ 
  In this paper, we will be interested in the {\em
    stability} of a given graph $G$. A graph $G$ is called {\em stable}
  ~\cite{bock2014finding} if 
  $\nu(G)=\tau_f(G)$. The class of
  such graphs properly contains {\em
    K\"onig-Egerv\'ary} graphs for which $\nu(G)=\tau(G)$ (e.g., see
  \cite{Kor82,KNP06,MR+11,Ste79}), which in turn properly contains the class of bipartite
  graphs~\cite{koniggrafok}. 
  The above unweighted stability notion generalizes
  naturally to weighted graphs. Given a graph $G=(V,E)$, edges weights
  $w \in \R^{E}_+$, we let $\nu(G,w)$ be the weight of a
  maximum-$w$-weight matching, and define
  \[ \tau_f(G,w) = \min \left\{ \sum_v y_v \,:\, y_u + y_v \geq w_{uv} ~
  \forall \{u,v\}  \in E, y \geq \0\right\}, \]
  to be the minimum cardinality of a fractional $w$-vertex cover. A
  weighted graph $(G,w)$ is then called {\em stable} if
  $\nu(G,w)=\tau_f(G,w)$. 

  Stable graphs are significant for the existence of \emph{good}
  solutions in certain combinatorial games. An instance of the
  matching game, introduced by Shapley and Shubik
  \cite{shapley1971assignment}, has a non-empty \emph{core} if and
  only if the underlying weighted graph is stable
  \cite{deng-ibaraki-nagamochi}. Moreover, an instance $(G,w)$ of a
  \emph{network bargaining game}, as introduced and investigated by
  Kleinberg and Tardos \cite{KT08}, admits a \emph{balanced solution}
  if and only if it admits a stable solution, i.e., if and only if
  $(G,w)$ is stable. Non-empty core and the existence of a balanced
  solution lead to \emph{fair benefit
  allocations} in assignment games and network bargaining games
  respectively.
}

\iffull{
\paragraph{K\"onig-Egerv\'ary graphs.}
A famous theorem of K\"onig \cite{koniggrafok} states that
$\nu(G)=\tau(G)$ holds whenever $G$ is bipartite. This result also
extends to the weighted setting as shown by Egerv\'ary
\cite{egervary1931matrixok}. That is, given edge weights
$w \in \mathbb{R}_+^{|E|}$, we know that the maximum $w$-weight of a
matching, denoted by $\nu(G,w)$, equals the minimum size of a
fractional $w$-vertex cover, denoted by $\tau_f(G,w)$. That is,
$\tau_f(G,w) = \min\{ \sum_{v\in V} y_v \mid y_u+y_v\ge w_{uv} \
\forall \{u,v\}\in E, \ y_v\ge 0\ \forall v\in V\}$.

In honor of K\"onig and Egerv\'ary, whose work built the foundations
for the development of Kuhn's Hungarian method
\cite{kuhn1955hungarian}, graphs satisfying $\nu(G)=\tau(G)$ are
called K\"onig-Egerv\'ary graphs (KEGs). The triangle extended by one
extra edge whose one end-vertex is a vertex of the triangle shows that
the class of KEGs properly extends the class of bipartite graphs. KEGs
are particularly significant since a minimum vertex cover can be
computed in KEGs in polynomial time \cite{evenK75}.

\paragraph{Stable graphs.}
A necessary condition for a graph $G$ to be a KEG is the existence of
a maximum integral matching whose size is equal to the value of a
minimum fractional vertex cover, i.e., $\nu(G)=\nu_f(G)=\tau_f(G)$.
Graphs satisfying this condition are called \emph{stable}
\cite{bock2014finding}. The family of stable graphs properly extends
the family of KEGs (e.g., consider the graph consisting of two
disjoint triangles linked by an edge that can be seen to be stable by
assigning $y_v=\nicefrac{1}{2}$ to every vertex $v$).

The notion of stable graphs naturally generalizes to the more general
setting with weights $w \in \mathbb{R}_+^{|E|}$ on the edges of
$G=(V,E)$. A weighted instance $(G,w)$ is called stable if and only if
the maximum $w$-weight of a matching equals the minimum size of a
fractional $w$-vertex cover, i.e., if and only if
$\nu(G,w)=\tau_f(G,w).$

\paragraph{Game-theoretic connections.}
Stable graphs are significant for the existence of \emph{good}
solutions to certain combinatorial games. An instance of the matching
game, introduced by Shapley and Shubik \cite{shapley1971assignment},
has a non-empty \emph{core} if and only if the underlying weighted
graph is stable \cite{deng-ibaraki-nagamochi}. Moreover, an instance
$(G,w)$ of a \emph{network bargaining game}, as introduced and
investigated by Kleinberg and Tardos \cite{KT08}, admits a
\emph{balanced solution} if and only if it admits a stable solution,
i.e., if and only if $(G,w)$ is stable. Non-empty core and the
existence of a balanced solution lead to \emph{fair benefit
  allocations} in assignment games and network bargaining games
respectively.

\paragraph{Alternative characterizations.}
We note that one can efficiently verify whether a given graph $G$ is
stable by computing a maximum weight matching using Edmonds' algorithm
and solving the minimum fractional weighted-vertex-cover linear
program. An alternative characterization of stability in unit-weight
graphs (all edges have the same weight) is given via the
\emph{Gallai-Edmonds decomposition} \cite{gallai1964maximale}. A graph
with unit weights is stable if and only if its inessential vertices,
i.e., those vertices not covered by at least one maximum cardinality
matching, form a stable set.
We observe that Edmonds' matching algorithm implicitly partitions the
vertex set into $V=X\cup Y\cup Z$ where $X$ denotes the inessential
vertices in $G$, $Y$ is the set of neighbours of $X$, and
$Z=V\setminus{(X\cup Y)}$ is the set of remaining vertices. Thus, the
terminating structure of Edmonds' matching algorithm can be used to
immediately test whether a graph $G$ is stable in polynomial time. We
recall that the partition $V=X\cup Y\cup Z$ is called the
\emph{Gallai-Edmonds decomposition} of $G$, and that $Y$ is also
called the \emph{Tutte set} of $G$. }
\fi

In \cite{bock2014finding}, Bock et al.\ proposed the following {\em
  meta} problem: given an {\em unstable} graph $G$, modify $G$ in the
least {\em intrusive} way in order to attain stability. The authors
focused on the concrete question of removing the smallest
number of edges from $G$ so that the resulting graph is stable. Bock
et al.\ showed that this problem is as hard to approximate as the
vertex cover problem, even if the underlying graph is factor critical
(i.e., even if deleting any vertex from $G$ yields a graph with a
perfect matching).
The authors complemented this negative result by presenting an
approximation algorithm whose performance guarantee is proportional to
the {\em sparsity} of the underlying graph. 

Concurrently, Ahmadian et al.~\cite{VertexDelStabilization16} and Ito
et al.~\cite{EdgeAddStabilization16} proposed a {\em vertex-stabilizer} problem: given a graph $G=(V,E)$, find a
minimum-cardinality set of vertices $S \subseteq V$ such that
$G[V\setminus S]$ is stable. Both papers presented a combinatorial
polynomial-time exact algorithm for this problem, and showed that
the min cost variants of vertex stabilization are NP-hard. Ito et
al.~\cite{EdgeAddStabilization16} proposed stabilizing a graph by
adding a minimum number of vertices or edges. They showed that
both of these problems are polynomial-time solvable. However, the minimum cost 
variant of stabilization by edge addition is NP-hard.

In this work, we consider a more {\em nimble} and in a sense {\em continuous} way
of stabilizing a given unstable graph $G=(V,E)$. Instead of
deleting/adding vertices/edges, we consider adding a small {\em
  subsidy} to a carefully chosen subset of the edges in order to
create a stable weighted graph. The subsidy should be thought of as an
additional incentive deployed by a central authority in order to
achieve stability. A natural goal for the central authority would then
be to minimize the total subsidy doled out in the stabilization
process. 

\begin{definition}[Minimum Fractional Additive Stabilizer]
  Given an undirected graph $G=(V,E)$ with unit edge weights, a
  \emph{fractional additive stabilizer} is a vector
  $c \in \mathbb{R}^{E}_+$ such that $(G,\1+c)$ is stable.
  In the \emph{minimum
    fractional additive stabilizer} (MFASP) problem, the goal is to
  find a fractional stabilizer of smallest weight $\1 ^Tc$. 
\end{definition}

We emphasize that we do not allow the addition of edges in MFASP,
but are restricted to add weight to existing edges. We further note
that the weight increases in MFASP need not be integral, and can
take on arbitrary non-negative rational values. 

\subsection{Our contributions.}\label{subsec:contribution} 

Several variants of graph stabilization are known to be
NP-hard. Hardness of approximation results so far have been rather
weak, however, and the gap between them and the known positive results
are large. In this work, we show strong approximation-hardness
results, and nearly matching positive results. 

\begin{restatable}{theorem}{thmHardness}\label{thm:hardness}
 A polynomial time approximation algorithm with approximation factor $O(|V|^{{1/24}-\epsilon})$ where $\epsilon>0$ for MFASP would lead to a polynomial
time $O(|V|^{{1/4}-6\epsilon})$-approximation for Densest $k$-Subgraph (D$k$S).
Furthermore, there is no $o(\log(|V|))$-approximation algorithm for MFASP unless $P = NP$.
\end{restatable}

D$k$S is known not to possess a polynomial-time approximation scheme,
assuming
$NP \not \subseteq \cap_{\epsilon > 0}
BPTIME(2^{n^\epsilon})$~\cite{KhotNoPTAS}.
On the other hand, the best known performance guarantee of any
approximation algorithm is only
$\approx O(|V|^{1/4})$~\cite{BhaskaraDkSG}.
It is widely believed, however, that the true approximability of DkS
lies closer to the upper bound than to the hardness lower-bound. 
An approximation algorithm for MFASP with performance ratio
significantly lower than
$|V|^{1/24}$ would therefore (at the very least) be unexpected. 

It is well-known that \eqref{lp:matching} has an integral solution iff
the set of {\em inessential} vertices $X$ (those vertices that are
exposed by a maximum matching) forms an independent set (e.g., see
\cite{Ba81,Uh75}). Let $Y$ be the set of neighbours of $X$ in $G$, and
$Z=V\setminus (X \cup Y)$. The triple $(X,Y,Z)$ is called the {\em
  Gallai-Edmonds decomposition} of $G$ \cite{edmonds1965maximum,Edmonds68,gallai1964maximale}.  

As we will see later, the optimization problem given by an instance of
MFASP naturally decomposes into two subproblems: that of picking a
maximum matching between the vertices in $Y$ and the factor critical
components in $G[X]$, and that of picking a maximum matching in each
of the components of $G[X]$. Our two hardness results in Theorem \ref{thm:hardness} demonstrate the hardness of each of these subproblems. 

In the following positive result, we let OPT denote the optimum
stabilization cost of the given instance.

\begin{restatable}{theorem}{thmOPTApprox}\label{thm:opt-approx}
  Let $G=(V,E)$ be a graph with Gallai-Edmonds decomposition
  $(X,Y,Z)$.  If all factor critical components of $G[X]$ have size
  greater than one then there is a
  $\min\{OPT, \sqrt{|V|}\}$-approximation algorithm for MFASP in $G$.
\end{restatable}

We note that the instances generated in the hardness proofs of Theorem
\ref{thm:hardness}
satisfy the properties needed in Theorem \ref{thm:opt-approx}. 

While stabilization by min-edge deletion is already NP-hard in factor-critical graphs \cite{bock2014finding}, we give a polynomial time algorithm to solve MFASP in factor-critical graphs. 

%
\begin{restatable}{theorem}{thmFactorCriticalGraphs}\label{theorem:FactorCriticalGraphs}
There exists a polynomial-time algorithm to solve MFASP in factor-critical graphs.
\end{restatable} 

We further exploit the efficient solvability of MFASP in factor-critical graphs to present an exact algorithm for MFASP in general graphs whose running time
is exponential only in the size of the Tutte set $Y$.  Thus, our
algorithm can be viewed as a fixed parameter algorithm (e.g., see
\cite{DowneyFellows}) where the parameter is the size of the Tutte
set. 

\begin{restatable}{theorem}{thmExactAlgo}\label{theorem:exact-algorithm-for-MFASP}
  There exists an algorithm to solve MFASP for a graph $G=(V,E)$ with
  Gallai-Edmonds decomposition $V=X\cup Y\cup Z$ in time
  $O(2^{|Y|}poly(|V|))$.
 \end{restatable}

We conclude by giving a conditional approximation algorithm that achieves
a $\nicefrac{(k+1)}{2}$-approximation when the number of non-trivial
factor-critical components in the Gallai-Edmonds-decomposition exceeds
the size of the Tutte set by a multiplicative factor of at least
$1+\nicefrac{1}{k}$. 


\iffull{
\begin{remark}
  In contrast to that of Bock et al.\ \cite{bock2014finding} our work
  considers an alternative method of stabilizing unstable instances.
  Bock et al.\ considered stabilizing through minimum edge-deletion
  and showed that it is NP-hard even in factor-critical graphs.  In
  contrast, we consider stabilizing by minimum fractional increase in
  edge weights. We are able to give an efficient algorithm for
  stabilizing factor-critical graphs through this method. Moreover, we
  obtain an efficient algorithm for arbitrary graphs with small sized
  Tutte set. In conjunction with the results of Bock et al., we have
  thus broadened the family of unstable instances that can be
  stabilized through efficient algorithms.
\end{remark}}

\subsection{Further related work.}

Various ways of modifying a given graph to achieve a property have
been studied in the literature, but most previous works seem to
consider {\em monotone} properties (e.g., see \cite{AS08,ASS05}). The
K\"onig-Egerv\'ary property is monotone while, notably, graph
stability is not. Most relevant to our work are the results of Mishra et al.\
\cite{Mishra07thecomplexity},
who studied the problem of finding a minimum number of edges to delete
to convert a given graph $G=(V,E)$ into a KEG. Akin to stable graphs,
KEGs are also significant in game theory: an instance of the
\emph{vertex cover game} has a non-empty \emph{core} if and only if
the underlying graph is a KEG \cite{deng-ibaraki-nagamochi}. Thus, in
the context of game theory, their study essentially addresses the
question of how to minimally modify an instance of a vertex cover game
so that the core becomes non-empty.  While they showed that it is
NP-hard to approximate the minimum edge-deletion problem to within a
factor of $2.88$, they also gave an algorithm to find a KEG subgraph
with at least $\nicefrac{3}{5}|E|$ edges.

In recent work, K\"{o}nemann et al.\ \cite{konemann2012network}
addressed a closely related problem of finding a minimum-cardinality
set of edges to remove from a graph $G$ such that the resulting graph
has a fractional vertex cover of value at most $\nu(G)$. We note that
the resulting graph here may not be stable.  While this problem is
known to be NP-hard \cite{BirX12}, K\"{o}nemann et al.\ gave an
efficient algorithm to find approximate solutions in sparse graphs.

\section{Preliminaries}\label{sec:prelims}

In the rest of the paper, we will only work with unit-weight graphs as
input instances for MFASP.  However, the results hold for uniform
weights since scaling preserves stability as well as our results.


We emphasize the following fact that is implicit from our earlier
discussion.  A graph $G$ is stable iff there is a maximum matching $M$
and $y \in \R^V_+$ such that the characteristic vector $\chi_M$ of $M$
and $y$ form an optimal pair of solutions for \eqref{lp:matching} and
\eqref{lp:cover}. A direct consequence of complementary slackness is
then that $y_v=0$ if $v$ is $M$-exposed, as well as $y_v+y_u=w_{uv}$
for all $\{u,v\} \in M$. A feasible solution to a MFASP instance
$G=(V,E)$ is determined by a triple $(M,y,c)$, where $M$ is
a matching, $y$ is a fractional $\1+c$-vertex
cover satisfying $\sum_{e\in M} 1+c_e=\sum_{v\in V} y_v$.
Moreover, such a matching $M$ will be a maximum $(1+c)$-weight
matching. Note that we use $w_e$ to refer to the total edge weight of an edge $e$, while $c_e$ to refer to the weight added for stabilizing.  

We recall the following properties of the Gallai-Edmonds decomposition (as defined in Section \ref{subsec:contribution}) (e.g., see \cite{Korte02,schrijver2003combinatorial}):  
Let $G=(V,E)$ and $V=X\cup Y\cup Z$ be the Gallai-Edmonds decomposition of $G$. Then
\begin{enumerate}[(i)]
\item every maximum matching in $G$ contains a perfect matching in $G[Z]$,
\item every connected component in $G[X]$ is factor-critical,
\item every maximum matching exposes at most one vertex in every connected component of $G[X]$, and
\item every maximum matching matches the vertices in $Y$ to distinct components of $G[X]$.
\end{enumerate}
We say that a component in $G[X]$ is non-trivial if it contains more than one vertex.

\section{Structural Results}\label{sec:properties-of-opt}

In this section, we show structural properties of optimal solutions to
MFASP which are useful to show hardness and design algorithms. The properties are summarized in the theorem below. 

\begin{restatable}{theorem}{thmoptproperties}
\label{thm:opt-properties}
Let $G=(V,E)$ be an instance of MFASP. Then, 
\begin{enumerate}[(i)] 
\item for every optimal solution $(M^*,y^*,c^*)$, 
\begin{enumerate}[(a)]
\item $c^*_e = 0$ for all edges $e \in E \setminus M^*$, $0 \leq c^*_e \leq 1$ for all edges $e \in M^*$, and
\item $|M^*| =\nu(G)$, i.e., $M^*$ is a maximum cardinality matching  in $G$.
\end{enumerate} 
\item there exists an optimal solution $(M^*, y^*, c^*)$ of MFASP with 
\begin{enumerate}
\item[(c)] half-integral $c^*$, and 
\item[(d)] $y^*\in \{0,\nicefrac{1}{2},1\}^{|V|}$ with the support of $y^*$ containing the Tutte set.
\end{enumerate}
\end{enumerate}
\end{restatable}

In the context of network bargaining games, the above structural
theorem (property (i)(b)) tells us that there exists a way to stabilize
through a minimum fractional additive stabilizer without
changing \emph{the number of deals} in the instance.\\

We split the proof of Theorem \ref{thm:opt-properties} into several Lemmas. Lemma \ref{lemma:opt-is-bounded} proves property (i)(a), Lemma \ref{lem:MFASP-max-matching} proves property (i)(b) and Lemmas \ref{lem:MFASP-half-integral} and \ref{lemma:PositiveTutteSet} prove property (ii). 

\begin{lemma}\label{lemma:opt-is-bounded}
Let $c$ be a minimum fractional additive stabilizer for a graph $G$. Let $M$ be a matching of maximum $(\1+c)$-weight.
Then $c_e = 0$ for all edges $e \in E \setminus M$ and $0 \leq c_e \leq 1$ for all edges $e \in M$.
\end{lemma}

\begin{proof}
Let $y$ be a minimal fractional $(\1+c)$-vertex cover. 
Since $c$ is a fractional stabilizer for $G$, by the discussion in Section \ref{sec:prelims}, 
it follows that $M$ and $y$ satisfy complementary slackness.

Let $\{u,v\}\in E\setminus M$.
If $c_{uv}>0$, then we may decrease $c_{uv}$ to zero: since $y$ is still a feasible fractional $(\1+c)$-vertex cover, and
$y$ satisfies complementary slackness with $M$, we obtain a better fractional additive stabilizer, thus contradicting the optimality of $c$.
Thus, for every edge $e\in E\setminus M$, we have $c_e=0$.

Let $\{u,v\}\in M$. Then by complementary slackness, we have that $y_u+y_v=1+c_{uv}$.
If $c_{uv}>1$, then we obtain $(c',y')$ where $c'_{uv}:=1$, $c'_e:=c_e$ for every edge $e\in E\setminus \{uv\}$ and $y'_u:=1,\ y'_v:=1$, $y'_i:=y_i$ for every vertex $i\in V\setminus \{u,v\}$.
The resulting solution $y'$ is a feasible fractional $(\1+c')$-vertex cover  and $y'$ satisfies complementary slackness with $M$.
Thus, $c'$ is a fractional additive stabilizer. We note that $\sum_{e\in E}c'_e<\sum_{e\in E} c_e$, a contradiction to the optimality of $c$.
\end{proof}

\begin{lemma}\label{lem:MFASP-max-matching}
For a graph $G$, let $c$ be a minimum fractional additive stabilizer. 
Then, the cardinality of a maximum $(\1+c)$-weight matching  is equal to the maximum cardinality of a matching in $G$.
\end{lemma}
\begin{proof}
Let $M$ be a maximum $(\1+c)$-weight matching. 
For the sake of contradiction, suppose the cardinality of $M$ is strictly less than the cardinality of a maximum matching in $G$. 
Let $y$ be a minimum fractional $(\1+c)$-vertex cover. 
Then, $M$ and $y$ satisfy complementary slackness.

Since, by our assumption, $M$ is not a maximum cardinality matching in $G$, there exists an $M$-augmenting path $P$. Let $u_{s}$ and $u_{e}$ denote the first and last vertices in the path $P$, respectively. Since $y$ is a minimal fractional $(\1+c)$-vertex cover, and $u_{s}$ and $u_{e}$ are exposed in $M$, we have
\begin{align}
y_u + y_v &\ge 1\ &\forall\ \{u,v\} \in P\setminus M,\\
y_u + y_v &=1+c_{uv}\ &\forall\ \{u,v\}\in M\cap P, \label{eq:compslack}\\
y_{u_{s}}&=0=y_{u_{e}}. & \label{eq:matching}
\end{align}
Let $N$ be the matching obtained by taking the symmetric difference of $M$ and $P$. Let us obtain new weights as follows:
\[
c'_{uv}:=
\begin{cases}
c_{uv}\ \text{ if } \{u,v\}\in E\setminus P\\
y_u+y_v-1\ \text{ if } \{u,v\}\in N\cap P\\
0\ \text{ if } \{u,v\}\in P\setminus N
\end{cases}
\]
We now show that the weight of matching $N$ w.r.t.\ $(1 + c')$ is identical to that of matching $M$:
\begin{align*}
\sum_{e\in N}(1+c_e') - \sum_{e\in M}(1+c_e)
&=  \sum_{\{u,v\}\in N\cap P}(1+(y_u+y_v-1)) - \sum_{ \{u,v\} \in M\cap P}(1+c_e) \\
&= \sum_{\{u,v\}\in M\cap P}(y_u + y_v)- \sum_{\{u,v\}\in M\cap P}(1+c_e)\\
&= \sum_{\{u,v\}\in M\cap P}(1+c_e) - \sum_{\{u,v\}\in M\cap P}(1+c_e)\\
&= 0.
\end{align*}
The second and third inequality are due to equations \eqref{eq:matching} and \eqref{eq:compslack}. 
By Definition of $c'$, 
we have that $y$ is a feasible  fractional $(\1+c')$-vertex cover in $G$.
Moreover, by Lemma \ref{lemma:opt-is-bounded} and the construction of $N$ and $c'$, 
the $(\1+c')$-weight of  matching $N$ is equal to the sum $\sum_{v\in V} y_v$ of the fractional $(\1+c')$-vertex cover $y$. 
Because of the LP duality relation between the two values, $N$ is a matching of maximum
$(\1+c')$-weight, and $y$ is a minimum fractional $(\1+c')$-vertex cover. Hence, $c'$ is a fractional additive stabilizer. Next we note that
\begin{align*}
\sum_{e\in E}c_e'-\sum_{e\in E}c_e
&= \sum_{\{u,v\}\in N\cap P} (y_u+y_v-1) - \sum_{\{u,v\}\in M\cap P} c_{ uv}\\
&= \sum_{\{u,v\}\in N\cap P} (y_u+y_v) - |N\cap P|- \sum_{\{u,v\}\in M\cap P} c_{uv}\\
&= \sum_{\{u,v\}\in M\cap P} (y_u+y_v) - |N\cap P|- \sum_{\{u,v\}\in M\cap P} c_{uv}\\
&= |M\cap P| - |N\cap P|\quad\\
&= -1.
\end{align*}
The third equality is due to \eqref{eq:matching}, the fourth follows form \eqref{eq:compslack}.
Hence, $c'$ is a fractional additive stabilizer whose weight is smaller than that of $c$, a contradiction to the optimality of $c$.
\end{proof}



\begin{lemma}
\label{lem:MFASP-half-integral}
For every graph $G$, there exists an optimal solution $(M,y,c)$ of MFASP with half-integral $c$ and half-integral $y$.
\end{lemma}
\begin{proof}
Let $\bar{c}$ be a minimum fractional additive stabilizer. By Lemma \ref{lem:MFASP-max-matching},
we know that there exists a maximum matching in $G$ that is also a maximum $(\1+\bar{c})$-weight matching.
Let $M$ be such a matching. We consider the following linear program:
\begin{align}
\min\sum_{e\in M}c_e \tag{LP(G,M)}\\
y_u + y_v &= c_{uv}+1\ &\forall\ \{u,v\} \in M \notag\\
y_u + y_v &\ge 1\ &\forall\ \{u,v\}\in E\setminus M \notag\\
y_u&=0\ &\forall\ u\in V,\ u \text{ is exposed by }M \notag\\
c,y&\ge 0 \notag
\end{align}

If $(c,y)$ is an optimal solution of $LP(G,M)$, then $c$ gives a minimum fractional additive stabilizer for $G$. 
In order to show that $c$ is a fractional additive stabilizer, 
it is sufficient to find a fractional $(\1+c)$-vertex cover $y$ that satisfies complementary slackness conditions with $M$. 
But, by the constraints in $LP(G,M)$, it is clear that $y$ satisfies complementary slackness conditions with $M$. 
Furthermore, $c$ is a minimum fractional additive stabilizer, since otherwise, we could derive a contradiction to the optimality of $\bar{c}$. 
Thus, it is sufficient to show that there exists a half-integral optimal solution $(c,y)$ of $LP(G,M)$.

We observe that if $G$ is bipartite, then for every matching $\bar{M}$ in $G$, 
the extreme point solutions to $LP(G,\bar{M})$ are integral since the constraint matrix is totally unimodular and the right-hand side is integral. 

Now, suppose $G=(V,E)$ is non-bipartite.
We construct a bipartite graph $G'=(V_1\cup V_2,E')$ as follows: for each vertex $u \in V$, we introduce vertices $u_1\in V_1$, $u_2\in V_2$ and for each edge $\{u,v\}\in E$,
we introduce edges $\{u_1,v_2\}, \{u_2,v_1\}$ in $E'$.
For each matching edge $\{u,v\}\in M$, we include edges $\{u_1,v_2\}, \{u_2,v_1\}$ in $M'$.
Thus $M'$ is a matching in $G'$ that exposes $u_1$ and $u_2$ for every vertex $u\in V$ that is exposed by $M$.
Let $(c',y')$ be an integral optimal solution of $LP(G',M')$.
Let $(c,y)$ be obtained by setting $c_{uv}:=\nicefrac{1}{2}(c'_{u_1v_2}+c'_{u_2v_1})\ \forall\ \{u,v\}\in M$ and 
$y_u:=\nicefrac{1}{2}(y'_{u_1}+y'_{u_2})\ \forall\ u\in V$.
Clearly, $(c,y)$ is half-integral. 
The following claim
shows that $(c,y)$ is an optimum to $LP(G,M)$.
\end{proof}

\begin{claim}\label{claim:optimality-of-transformed-solution}
Let $(c',y')$ be an optimal solution of  $LP(G',M')$. Then $(c, y)$ obtained by setting
$c_{uv} := \nicefrac{1}{2}(c'_{u_1v_2}+c'_{u_2v_1})$ for all $\{u,v\}\in M$ and $y_u:=\nicefrac{1}{2}(y'_{u_1}+y'_{u_2})$ for all $u\in V$
is an optimal solution for $LP(G, M)$.
\end{claim}
\begin{proof}
The feasibility of the solution $(c,y)$ for $LP(G,M)$ is easy to verify.
We note that $\sum_{e\in M}c_e=\nicefrac{1}{2}\sum_{e\in M'}c'_e$. We will prove optimality.

Suppose $(c,y)$ is not optimal for $LP(G,M)$. Then there exist $(\tilde{c},\tilde{y})$  feasible for $LP(G,M)$
such that $\sum_{e\in M}\tilde{c}_e<\sum_{e\in M}c_e$.
Consider the solution $(\tilde{c}',\tilde{y}')$ obtained by setting $\tilde{c}'_{u_1v_2}=\tilde{c}'_{u_2v_1}=\tilde{c}_{uv}$
for every $\{u,v\}\in M$ and $\tilde{y}'_{u_1}=\tilde{y}'_{u_2}=\tilde{y}_u$ for every $u\in V$.
The resulting solution $(\tilde{c}',\tilde{y}')$ is feasible to $LP(G',M')$.
Moreover $\sum_{e\in M'}\tilde{c}_e'=2\sum_{e\in M}\tilde{c}_e<2\sum_{e\in M}c_e=\sum_{e\in M'}c'_e$,
a contradiction to the optimality of $(c',y')$.
\end{proof}

\begin{lemma}
\label{lemma:PositiveTutteSet}
For every graph $G$, there exists an optimal solution $(M,y,c)$ of MFASP with half-integral $c$ 
and half-integral $y$ with the support of $y$ containing the Tutte set.
\end{lemma}
\begin{proof}
Let the Gallai-Edmonds decomposition of $G$ be given by $V = X \cup Y \cup Z$. Let $c$ be a half-integral minimum fractional additive stabilizer for $G$. 
Let $M$ be a maximum $(\1+c)$-weight matching and $y$ be a half-integral minimum fractional $(\1+c)$-vertex cover 
(such a $y$ and $c$ exist by Lemma \ref{lem:MFASP-half-integral}). 
Suppose that $y_v = 0$ for some $v \in Y$. 
We will construct a half-integral fractional additive stabilizer $c'$ without increasing the cost and a fractional $(\1+c')$-vertex cover $y'$ 
that satisfies complementary slackness with $M$ and has $y'_w>0$ for each node $w$ of the Tutte set. 

Since $M$ is maximum, every node of $Y$ is matched. For $v \in Y$, we denote by $S_v$ the factor-critical component in $G[X]$ which is matched to $v$ and by $s_v$ the vertex matched to $v$. Let $Y' := \{v \in Y : y_v = 0\}$. We set $c'_e := 0 \ \forall e\in \bigcup_{v \in Y'} (E(S_v) \cup \{v, s_v\})$ and $c'_e := c_e$ otherwise. It is clear that $c'$ is half-integral and the cost of $c'$ cannot be more than that of $c$.

We define 
\[
y'_w := \begin{cases}
	1/2,  & \text{if } w \in Y' \text{ or } w \in \bigcup_{v \in Y'} V(S_v),\\
	y_w & \text{else}.
	\end{cases}
\]
By definition, $y'$ satisfies the covering constraints in $\tau_f(G,\1+c')$ for edges in  
$E[\bigcup_{v \in Y'}(S_v) \cup Y']$. For other edges $\{v, t\}$ incident to $v \in Y' \bigcup _{v \in Y'}(S_v)$
 either $y'_t = y_t = 1$ (if $v \in Y'$) or $y'_t \ge  \nicefrac{1}{2}$ (if $t \in Y \setminus Y'$).  
Therefore $y'$ is a fractional $(\1+c')$-vertex cover. Finally it is clear that $y'$ satisfies complementary slackness with $M$: on every matching edge that is not adjacent to a vertex $v\in Y'$, it follows since $y'$ takes the same values on the end vertices as $y$; for a matching edge $\{u,v\}$ adjacent to a vertex $v\in Y'$, by definition of $y'$ and $c'$, it follows that $y_u+y_v=1+c_{uv}$. 
\end{proof}


Using Theorem \ref{thm:opt-properties}, we will always use
and construct solutions where $y_v \ge \nicefrac{1}{2}$ for a vertex
$v$ in the Tutte set in the remaining paper. Therefore, we assume in the subsequent sections, that
$Z = \emptyset$ in the Gallai-Edmonds decomposition of a graph. If
that were not the case, we would first consider the graph without $Z$
and then extend the stabilizer using a perfect matching on $Z$ without
additional cost. This is done by setting $c_e=0$ for every edge $e=\{u,v\}\in E[Z]\cup \delta(Z)$ and $y_v = \nicefrac{1}{2}$ for all $v \in Z$.

We now use the structural insights from Theorem \ref{thm:opt-properties} to describe the behaviour of feasible solutions on the factor-critical components of the Gallai-Edmonds decomposition. 

\begin{restatable}{lemma}{lemmaFactorCriticalNotFree}\label{lemma:FactorCriticalNotForFree}
Let $G$ be a graph with Gallai-Edmonds decomposition $V(G) = X\cup Y \cup Z$ and $(M, y, c)$ be a feasible solution for MFASP in 
$G$ fulfilling the properties (a), (b), (c) and (d) of Theorem \ref{thm:opt-properties}. 
Let $K$ be a non-trivial component in $G[X]$ with a vertex $u$ such that $y_{u}=0$. 
If $K$ has a vertex $u$ exposed by $M$, then $\sum_{e\in E(K)} c_e  \geq 1$.
On the other hand, if $K$ is matched to $Y$ by an edge $e' = \{v, w\}$ with $w \in Y$ and $y_w \ge \nicefrac{1}{2}$, 
then $\sum_{e\in E(K)} c_e + c_{e'} \geq y_w$. 
\end{restatable}

\begin{proof}

In this proof, we use an equivalent definition of factor-critical graphs: 
A graph is factor-critical if and only if it has an odd ear-decomposition. 
Furthermore, the initial vertex of the ear-decomposition can be chosen arbitrarily \cite{Lovasz72}. 
 An ear-decomposition of a graph $G$ is a sequence $r, P_1, \ldots, P_k$ with $G = (\{r\}, \emptyset) + P_1 + \ldots + P_k$
 such that $P_i$ is either a path where exactly the endpoints belong to $\{r\} \cup V(P_1) \cup \ldots V(P_{i-1})$ or a circuit 
 where exactly one of its vertices belongs to $\{r\} \cup V(P_1) \cup \ldots V(P_{i-1})$. An ear-decomposition is called odd if all $P_i$ have odd length.

Moreover, (the proof of the above equivalence implies that) for a maximum matching $M$ in a factor-critical graph $G$, an odd ear-decomposition can be chosen
such that the exposed vertex is the initial vertex $r$, 
each path $P_i$ is $M$-alternating such that the first and last edge are not part of $M$ 
and each circuit $P_i$ contains $\nicefrac{| E(P_i) - 1|}{2}$ matching edges such that the 
vertex in  $\{r\} \cup V(P_1) \cup \ldots V(P_{i-1})$ is not matched. 

If $u$ is exposed by $M$, then $y_u=0$. 
Since $K$ is factor-critical, there exists an $M$-blossom through $u$, i.e.\ a circuit $C$ of odd length where all vertices except $u$ are 
adjacent to an edge in $E(C) \cap M$ (follows from the above statement). Let the vertices of $C$ be $u_0 = u, u_1, \ldots, u_{2t}, u_{2t + 1} = u$. 
Now consider the optimal vertex cover values $y_{u_0},\ldots,y_{u_{2t}}$ for the vertices in $C$. By definition of the vertex cover, $y_{u_i}+y_{u_{i+1}}\ge 1$ for every $i=0, \ldots ,2t$ and, in particular, $y_{u_1} \geq 1$ and $y_{u_{2t}} \geq 1$.
Furthermore, the inequalities for the matching edges are tight and thus, $1+c_{\{u_i,u_{i+1}\}} = y_{u_i}+y_{u_{i+1}}$ for every $i= 1,3,5, \ldots ,2t-1$. 
Therefore, summing up, we have 
\begin{align*}
t + \sum_{e\in E(K)} c_e &\geq t + \sum_{i=1,3,5,\ldots,2t-1} c_{\{u_i,u_{i+1}\}}
= \sum_{i=1}^{2t} y_{u_i}
\geq 1+ \left(\sum_{i=2}^{2t-1} y_{u_i}\right) + 1  \\
&= 2 + \sum_{i=2,4,\ldots,2t-2} (y_{u_i} + y_{u_{i+1}})
\geq 2 + (t-1) = t+1, 
\end{align*}
 which proves the first statement.
 
Now, let us consider the case where $K$ is matched. 
Remember that we have $y_w \in \{\nicefrac{1}{2}, 1\}$. 
If $y_v = 0$, then the proof is identical to that of the first statement. 
If $y_v = 1$, clearly, $c_{e'} = y_w$. Therefore, we may assume that $y_v = \nicefrac{1}{2}$
and thus $c_{e'} = y_w - \nicefrac{1}{2}$. It remains to show that $\sum_{e\in E(K)}c_e\ge 1/2$. 

Since $K$ is factor-critical, there exists a path from $u$ to $v$ of odd length in $K$, which is $M$-alternating, in particular, 
the edges incident to $u$ and $v$ are non-matching edges. 
(The existence of such a path follows from the fact that it is possible to construct an odd ear-decomposition of $K$ with 
initial vertex $v$ such that each ear is an $M$-alternating path or circuit where the first and last edge are not matching edges.) 
Let $u = v_0, v_1, \ldots, v_{2t + 1} = v$ be such a path $P$. 
Suppose for the sake of contradiction $c_e = 0$ for all $e \in E(K)\supseteq P$. Since $y$ is a vertex cover and the matching edges are tight w.r.t.\ $y$, 
it follows that $y_{v_i} = 1$ for odd $i \leq 2t $ and $y_{v_i} = 0$ for even $i$. 
But that implies $y_{v_{2t}} + y_v = y_v = \nicefrac{1}{2}$, a contradiction. Hence, $\sum_{e\in E(K)}c_e>0$. 
Since $c$ is half-integral, it follows that $\sum_{e\in E(P)}c_e \geq \nicefrac{1}{2}$. 
\end{proof}

This directly implies the following statement: 

\begin{proposition}\label{lemma:nontrivialproperties}
  Let $G$ be a graph with Gallai-Edmonds decomposition
  $V(G) = X\cup Y \cup Z$ and $(M, y,c)$ be a (not necessarily
  optimum) solution for MFASP in $G$ fulfilling the properties (a), (b), (c) and (d) of
  Theorem \ref{thm:opt-properties} such that $(y, c) = \text{argmin}\{1^Tc:(M, y, c) \text{ is feasible for MFASP}\} $. 
  Let $K$ be a non-trivial factor-critical component in
  $G[X]$.  If $K$ is matched by $M$, then $y_v=1/2$ for every
  $v\in V(K)$ and $c(e)=0$ for every $e\in K$.
\end{proposition}

\section[Hardness]{Inapproximability}\label{sec:hardness}
In this section, we will show that MFASP is hard to approximate in general graphs. 
We show the first part of Theorem \ref{thm:hardness} in Section \ref{sec:dksg} and the second part in Section \ref{sec:set_cover}. 

\subsection{Reduction from Densest k-Subgraph}
\label{sec:dksg}
In this subsection, we show that MFASP is at least as hard as the Densest $k$-Subgraph Problem in a certain approximation preserving sense. 
In Theorem \ref{prp:MkEChardness}, we show that a polynomial time $f$-approximation algorithm for the MFASP would imply a polynomial time $2f$-approximation for the Minimum $k$-Edge Coverage Problem. We show that Theorem \ref{prp:MkEChardness} implies a strong relation to the 
Densest $k$-Subgraph Problem at the end of the subsection.

We recall the two problems of relevance: Given a graph and a parameter $k$, the Minimum $k$-Edge Coverage Problem (M$k$EC)
is to find a minimum number of vertices whose induced subgraph has at least $k$ edges. The Densest $k$-Subgraph Problem asks for $k$ vertices with a
maximum number of induced edges.

\begin{theorem}\label{prp:MkEChardness}
If there is a polynomial time $f$-approximation algorithm for the MFASP, then there is a polynomial
time $2f$-approximation for M$k$EC.
\end{theorem}

\begin{proof}
  Let $G = (V, E)$ and $k < |E|$ be an instance of M$k$EC. Our goal is
  to find a subset of $k$ edges spanning a minimum number of nodes.
  We construct an instance $\hat{G}$ of MFASP
  whose Gallai-Edmonds
  decomposition (GED) has a specific form.  By Theorem
  \ref{thm:opt-properties}, it is sufficient to consider maximum
  matchings for MFASP.  $\hat{G}$ will encode the problem of picking
  $k$ edges for M$k$EC as the problem of identifying $k$
  factor-critical components in the GED that are to be exposed by the
  matching in a solution for MFASP. An illustration can be found in Figure \ref{fig:MkEChardness}.

  Let $Y'$ be a copy of vertex set $V$; we will later show that $Y'$
  is a part of the Tutte set $Y$ of the GED of the constructed graph.
  Furthermore, for each edge $e = \{v, w\} \in E$, we add a triangle,
  and we let $\Delta$ denote the collection of these triangles.  We
  will later show that each triangle will form a component of $X$ in
  the GED.  We connect each node of a triangle corresponding to an
  edge $\{v, w\} \in E$ to the vertices $v$ and $w$ in $Y'$.  As any
  maximum matching matches each vertex of the Tutte set to a distinct
  factor-critical component, we modify the instance such that the
  number of triangles is exactly $|V| + k$.  To achieve this, we
  either add vertices to $Y'$ which are connected to all vertices of
  all triangles or we add triangles that are connected to all vertices
  in $Y'$.

  While M$k$EC allows choosing any collection of $k$ edges,
  there may exist a collection of $k$ triangles in our current graph such
  that the remaining triangles cannot be matched perfectly to $Y'$.
  To remedy the situation, we add $q - 1$ copies
  $Y_2, \ldots Y_q$ of $Y'$, where $q$ will be chosen later, and connect each vertex of $Y_i$
  ($i \in 2, \ldots q$) with the same nodes as the corresponding node
  in the ``original'' set $Y'$, which can be seen as $Y_1$.  We will
  later show that all these copies belong to the Tutte set $Y$ of the
  Gallai-Edmonds decomposition.  Moreover, we add $|Y'|*(q - 1)$
  triangles and connect all their vertices to all vertices in
  $Y'' := Y_1\cup \ldots \cup Y_q$.  Call this set of newly added
  triangles $C'$.

      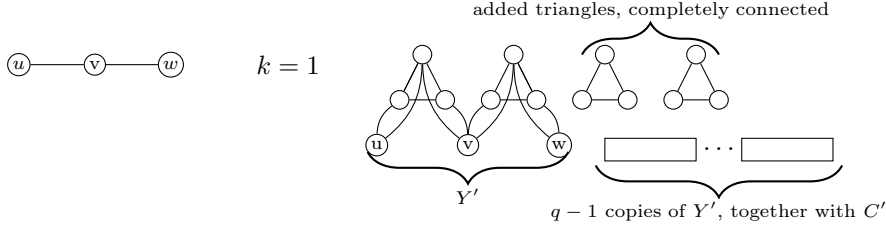
\begin{figure}[tb]

     \tikzstyle{vertex}=[circle, draw,
                        inner sep=1pt, minimum width=7pt]
    \tikzstyle{edge} = [draw,thick]
    \centering
\begin{subfigure} [M$k$EC instance $G$]{0.35\linewidth}
{
      \centering
    \begin{tikzpicture} [scale = 0.5]

    \foreach \x/\nr/\name in {{0/1/u}, {4/3/w}}
        \node[vertex] (v\nr) at (\x, 0){\scriptsize $\name$};

     \node[vertex] (v2) at (2, 0){\scriptsize v};

    \foreach \i/\j in {{2/3}, {1/2}}
    {
        \draw (v\i) edge (v\j);
    }

     \draw [white] (2,-2) circle (1mm);
     \draw (6, 0) node[anchor = west]{$k = 1$};
    \end{tikzpicture}
} \end{subfigure}
\quad
\begin{subfigure} [constructed $\hat{G}$]{0.6\linewidth}
{
      \centering
    \begin{tikzpicture}[scale = 0.3]
    \foreach \nr/\coord/\name in {{1/0/u}, {2/4/v}, {3/8/w}}
    {
        \node[vertex] (v\nr) at (\coord, 0){\scriptsize \name};
    }

    \foreach \i/\j/\coord in {{v/w/6}, {u/v/2}, {w/a/10}, {x/a/14}}
    {

        \node[vertex] (r\i1) at (\coord - 1, 2){};
        \node[vertex] (r\i2) at (\coord + 1 , 2){};
         \node[vertex] (r\i3) at (\coord , 4){};
         \foreach \k/\l in {{1/2}, {2/3}, {3/1}}
         {
         \draw (r\i\k) edge (r\i\l);
         }
    }

    \foreach \i/\j in {{v1/ru1}, {v2/ru3}, {v2/rv1}, {v3/rv3}}
    {
        \draw (\i) edge [bend left = 25] (\j);
    }

    \foreach \i/\j in {{v1/ru3}, {v2/ru2}, {v2/rv3}, {v3/rv2}}
    {
        \draw (\i) edge [bend right = 25](\j);
    }

    \draw (ru1) edge[] node[left]{}(ru3);
    \draw (rv2) edge[] node[right]{}(rv3);

    \draw[thick, decorate,decoration={brace,amplitude=12pt}] (8.4,-0.3) -- (-0.4,-0.3) node[midway, below ,yshift=-10pt,]{\scriptsize $Y'$};

    \draw[thick, decorate,decoration={brace,amplitude=12pt}] (9,4) -- (15,4) node[midway, above ,yshift=10pt,]{\scriptsize added triangles, completely connected};

    \draw (10,-0.7) rectangle (14, 0.3);

    \draw (15, 0.3) node[anchor = north]{$\ldots$};

     \draw (16,-0.7) rectangle (20, 0.3);

    \draw[thick, decorate,decoration={brace,amplitude=12pt}] (20.4,-0.9) -- (9.6,-0.9) node[midway, below ,yshift=-10pt,]{\scriptsize $q - 1$ copies of $Y'$, together with $C'$};

    \end{tikzpicture}
} \end{subfigure}

\caption{Instance of M$k$EC and schema of a corresponding MFASP instance $\hat{G}$}
\label{fig:MkEChardness}
\end{figure}

The following two lemmas describe the relevant structure of the construction.


\begin{lemma}\label{lemma:perfectrestmatching}
 Let $q \ge \max_{v \in V(G)} |\delta(v)|$.
 Then for any choice of $k$ triangles, there is a perfect matching between $Y''$ and the triangles that were not chosen.
\end{lemma}

\begin{proof}
  Let ${\cal E}$ be the set of $k$ triangles that we wish to expose.
  We construct a matching $M_{\cal E}$ that exposes precisely ${\cal
    E}$. For each triangle in $\Delta\setminus {\cal E}$ corresponding
  to some edge $\{u,v\} \in E$ match the triangle to a currently
  exposed copy of $u$. Note that $q$ is at least the maximum degree in
  $G$, and hence this process matches all triangles in
  $\Delta\setminus {\cal E}$.

  Let $\widetilde{\Delta}$ be the collection of triangles not in $\Delta\cup {\cal E}$, and let
  $\widetilde{Y''}$ be the collection of $M_{\cal E}$ exposed vertices in
  $Y''$. Clearly, $|\widetilde{\Delta}|=|\widetilde{Y''}|$, and the graph
  induced by the edges between $\widetilde{Y''}$ and the vertices of
  triangles in $\widetilde{\Delta}$ is complete bipartite. Thus, picking
  any $\widetilde{Y''}$-perfect matching in this graph and adding its edges to $M_{\cal E}$
  yields the desired matching exposing ${\cal E}$.
\end{proof}

\begin{lemma}\label{lemma:MkECGEdecomposition}
 The Gallai-Edmonds decomposition of $\hat{G}$ is given by $Y = Y''$, $X = V(\hat{G})\setminus Y''$, $Z = \emptyset$.
\end{lemma}

\begin{proof}
In any graph $G=(V,E)$, the size of a maximum matching can be characterized by the \emph{Tutte-Berge formula} \cite{Berge1958}:
$$2\nu(G) = |V| - \max_{W \subseteq V} (q_G(W) - |W|)$$ 
where $q_G(W)$ denotes the number of
components with an odd number of nodes in $G[V\setminus{W}].$

By plugging $W = Y''$ into the Tutte-Berge formula, we see that a maximum matching has size at most $\frac{V(\hat{G}) - k}{2} = 2|Y''| + k$. A matching of size $2|Y''|+k$ exists by Lemma \ref{lemma:perfectrestmatching}.

By Lemma \ref{lemma:perfectrestmatching}, every vertex $v \in  V(\hat{G})\setminus  Y''$ is inessential, that is, there exists a maximum matching in $\hat{G}$ exposing $v$.
Now, suppose there was a maximum matching $M$ exposing a vertex $v \in Y''$.
We know $|M| = 2|Y''| + k$, but any matching can contain at most $|Y''| + k$ edges of $E[V(\hat{G})\setminus (Y'')]$, as that is the number of triangles.
All other edges have one endpoint in $Y''$. Thus, if $M$ exposes $v$, then $|M| < |Y''| + k + |Y''|$, which is a contradiction.
\end{proof}

Together, Lemmas \ref{lemma:perfectrestmatching} and \ref{lemma:MkECGEdecomposition}  imply that for any choice of
$k$ unmatched factor-critical components, there is a maximum matching exposing exactly one vertex in these $k$ components
and conversely, every maximum matching is of this form.
We have shown in Theorem \ref{thm:opt-properties} that it suffices to consider stabilizers $(M, y, c)$ where $M$ is a maximum matching,
$y$ and $c$ are half-integral and $y$ is positive on the Tutte set.
Once the set of unmatched components is fixed, we can see how to obtain an optimal stabilizer for this situation:
Start with a matching between the matched components and $Y''$ and extend it arbitrarily to a maximum matching $M$.
Let $K \subset V(\hat{G})$ be the set of vertices in triangles not matched to $Y''$.
We set $c_e = 1$ for the matching edge in each of these triangles, $y_v = 1$ for both matched vertices within the triangle and $y_v = 0$ for the remaining vertices in $K$.
For $v \in Y''$, we set $y_v = 1$ if $v \in N(K)$ and $y_v = \nicefrac{1}{2}$ otherwise.
By Proposition \ref{lemma:nontrivialproperties}, $y_v = \nicefrac{1}{2}$ is then optimal for all vertices $v$ in matched triangles.
Consequently, $c(e) = \nicefrac{1}{2}$ if $e \in M \cap \delta(N(K))$ and $c(e) = 0$ for the remaining edges.
In total, $\1^Tc = k + \nicefrac{1}{2}|N(K)|$.

If the unmatched components correspond to a set $E'$ of edges in $G$,
then $N(K) \cap Y_1$ corresponds exactly to the vertices in $G$ spanned by $E'$.
Consequently, the cost of the stabilizer consists of $k$ (for the unmatched triangles) and $q$ times
the number of spanned vertices in $G$ as the neighbourhoods of the copies of $Y'$ are identical.
Suppose a component that does not correspond to an edge in $G$ is unmatched. Then, $y_v = 1$ for all $v \in Y''$
and therefore $c(E) = k + \nicefrac{1}{2}|Y''|$ .
Thus, we can modify the solution by choosing to expose components corresponding to edges instead without increasing the cost.
W.l.o.g.\ we modify any solution to MFASP this way. Then, we have the following Lemma:

\begin{lemma}\label{lemma:mkec-mfasp}
$G$ has a solution of M$k$EC of size at most $x$ if and only if $\hat{G}$ has a MFASP of cost at most $k+qx/2$.
\end{lemma}

We next show that Lemma \ref{lemma:mkec-mfasp} yields a factor-preserving hardness. If $k > 0$, then we have $x \ge 2$. Moreover, set $q = max\{k, \max_{v \in G} |\delta(v)|\}$.
Let $x^*$ be the value of an optimal solution for M$k$EC, then the optimal value of MFASP is $k + \frac{qx^*}{2}$.
Suppose there was an $f$-approximation for MFASP, this would yield a stabilizer solution of cost $k + \frac{qx}{2} \le f( k + \frac{qx^*}{2})$ for some $x$. We observe that $\frac{qx}{2} \le (f - 1)k + f\frac{qx^*}{2} \le qf(1 + \frac{x^*}{2}) \le qfx^*$. Therefore, we have a $2f$-approximate solution of M$k$EC
which proves Theorem \ref{prp:MkEChardness}.
\end{proof}

We now show the first part of Theorem \ref{thm:hardness}. While using Theorem \ref{prp:MkEChardness} to derive hardness of approximation for MFASP, we have to be careful if we want to set $f$ to be a value that depends on the input size: If $G$ is the input for M$k$EC with $n = |V(G)|$, 
then the number of vertices in our construction of the MFASP is bounded by $\max\{4n^3 + n^2, 7n^2\} \leq 7n^3$. For $\hat{n}$ being the number of vertices in an MFASP instance, we can conclude that an approximation algorithm with approximation factor $O(\hat{n}^{\frac{1}{24}-\epsilon})$, for $\epsilon>0$ would lead to a $O(n^{1/8-3\epsilon})$ approximation algorithm for M$k$EC, where $n$ is the number of vertices of the M$k$EC instance. This would lead to an algorithm with approximation factor $O(n^{1/4-6\epsilon})$ for M$k$DS according to \cite{HajiJain06}. 

\subsection{Reduction from Set Cover}
\label{sec:set_cover}


While the Densest-$k$-subgraph problem is believed to be difficult, there are no strong inapproximability results known.
In this subsection, we 
show Set-Cover-hardness for MFASP, which leads to a stronger inapproximability result.

We exploit a different aspect of MFASP for this reduction: We could look at MFASP as a problem consisting of two subproblems:
How to choose the matched factor-critical components and, having fixed those, how to choose the matching within the unmatched components
and thus decide the $y$-values.
In the previous reduction, the difficulty was completely encoded in the first subproblem. Once we chose the matched components,
the second problem was easy.
In the following reduction, we will consider a construction, where the matched components are the same for any reasonably good solution and the
difficulty lies in the second subproblem.

\begin{theorem}
 \label{thm:MFASP-SC-hardness}
 If there is a polynomial time $f$-approximation algorithm for MFASP,
 then there is a polynomial time $2f$-approximation algorithm for Set Cover.
\end{theorem}

\begin{proof}
 Let $(\mathcal{S},\mathcal{X})$ be an instance of the Set Cover problem with sets $\mathcal{S}=\{S_1,\dots,S_m\}$ and elements $\mathcal{X}=\{x_1,\dots,x_n\}$. Our goal is to choose a minimum number of sets whose union contains all elements $x_i$. Let the frequency of element $x_i$ be $F_i = |\{ S_j : x_i \in S_j \}|$. Without loss of generality, $F_i > 1$. Otherwise, the only set containing an item $x_i$ has to be part of any solution, so it suffices to consider instances with $F_i>1$ for all $i\in [n]$.

We construct a graph $\hat{G}$ with a specific Gallai-Edmonds decomposition. 
Our goal will be to decide whether a set is included in a set cover based based on the $y$-values of the Tutte set. 
For each set $S_j$, create $n$ vertices $S_j^1, \ldots, S_j^n$. Let $Y'=\{S_j^i: i\in [n],j\in[m]\}$. This will be our Tutte set.
For each $S_j^i$ create a clique $C_j^i$ of size $2N+1$ with $N>(nm)^2$ with a designated vertex $c_j^i$ and add an edge $\{S_j^i, c_j^i\}$.
The purpose of these large cliques is ensuring that every vertex in $Y'$ is matched to its clique, thus exposing the factor-critical components we construct next:
For each element $x_i$ with $F_i$ odd, construct an odd cycle $Q_i$ consisting of $F_i$ vertices $x_i^1,\ldots, x_i^{F_i}$.
For each element with $F_i$ even, construct an odd cycle $Q_i$ consisting of $F_i+1$ vertices $x_i^1,\ldots, x_i^{F_i+1}$,
where the vertex $x_i^{F_i+1}$ is a dummy vertex.
Let $\hat{S}_{(1,i)},\ldots, \hat{S}_{(F_i,i)}$ denote the sets in $\mathcal{S}$ containing $x_i$ (choose the order arbitrarily).
Consider the $n$ copies of the corresponding vertices in $Y'$ and
add edges $\{x_i^k, \hat{S}_{(k,i)}^{\ell}\}$ $\forall\ \ell\in [n]$, $\forall\ k\in [F_i]$ $\forall\ i\in [n]$.
I.e.\ add an edge between the $k$-th vertex for $x_i$ and all copies of the $k$-th set in the list.
For every $i\in [n]$ with $F_i$ even, add edges between $x_i^{F_i+1}$ and all vertices in $Y'$.
Let the resulting graph be $\hat{G}=(\hat{V},\hat{E})$. (See figure \ref{fig:SetCoverHardness}).

     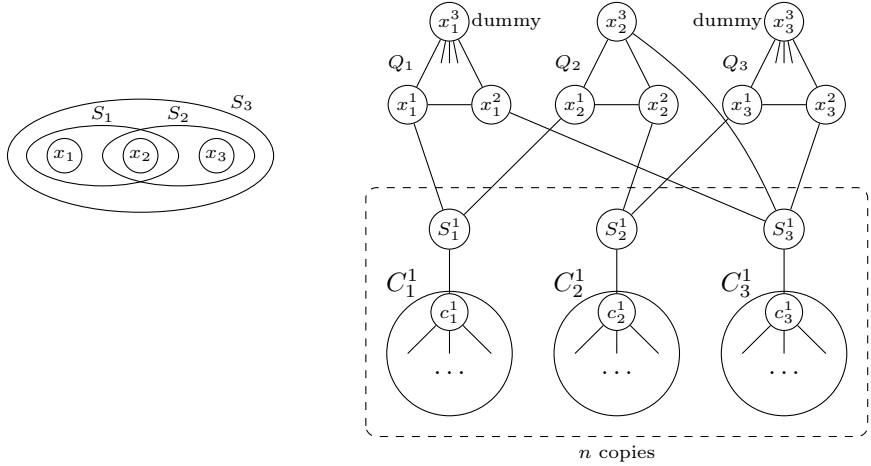
\begin{figure}[tb]

      \tikzstyle{vertex}=[circle, draw,
                         inner sep=1pt, minimum width=13pt]
     \tikzstyle{edge} = [draw,thick]
     \centering
 \begin{subfigure}[set cover instance]{0.35\linewidth}
 {
       \centering
     \begin{tikzpicture}[scale = 0.5]

     \foreach \x/\nr/\name in {{0/1/1}, {2/2/2}, {4/3/3}}
         \node[vertex] (v\nr) at (\x, 0){\scriptsize $x_\name$};

    	\draw (1, 0) ellipse  (2cm and 0.8cm) node[below, yshift = 22 pt]{\scriptsize $S_1$} ;
 	 \draw (3, 0) ellipse (2cm and 0.8cm) node[below, yshift = 22 pt]{\scriptsize $S_2$};
	 \draw (2, 0) ellipse (3.5cm and 1.5cm) node[right, xshift = 30 pt, yshift = 20 pt]{\scriptsize $S_3$};

      \draw [white] (2,-6) circle (1mm);

     \end{tikzpicture}
 }\end{subfigure}
 \quad
 \begin{subfigure}[constructed $\hat{G}$]{0.6\linewidth}
 {
       \centering
     \begin{tikzpicture}[scale = 0.55]

     \foreach \nr/\coord/\name in {{1/0/1}, {2/4/2}, {3/8/3}}
     {
         \node[vertex] (v\nr) at (\coord, 0){{\scriptsize $S_\name^1$}};
         \node[vertex] (c\nr) at (\coord, -2){{\scriptsize $c_{\name}^1$}};
         \draw (v\nr) edge (c\nr);
         \draw (c\nr) edge (\coord -1 , -3);
         \draw (c\nr) edge (\coord  , -3);
         \draw (c\nr) edge (\coord +1 , -3);
         \draw (\coord, -3.2) node[anchor = north]{$\ldots$};
         \draw [] (\coord,-3) circle (1.5 cm);
         \draw (\coord - 0.5, -1.3) node[anchor = east]{$C_{\name}^1$};
     }

     \draw [dashed, rounded corners] (-2, -5) rectangle (10, 1);
     \draw (4, -5) node[below]{\scriptsize $n$ copies};

%
     \foreach \i/\coord in {{1/0}, {2/4}, {3/8}}
     {
        \node[vertex] (r\i1) at (\coord - 1, 3){\scriptsize $x_\i^1$};
         \node[vertex] (r\i2) at (\coord + 1 , 3){\scriptsize $x_\i^2$};
          \node[vertex] (r\i3) at (\coord , 5){\scriptsize $x_\i^3$};
          \foreach \k/\l in {{1/2}, {2/3}, {3/1}}
          {
          \draw (r\i\k) edge (r\i\l);
          }
          \draw (\coord - 0.6, 4) node[anchor = east]{\scriptsize $Q_\i$};
     }
      \foreach \i/\j in {{r11/v1}, {r12/v3}, {r21/v1}, {r22/v2}, {r31/v2}, {r32/v3}}
      		\draw (\i) edge (\j);

        \foreach \i/\j in {{r23/v3}}
       		\draw (\i) edge [bend left = 15] (\j);

      \draw (0.3, 5) node[right]{\scriptsize dummy};
        \draw (7.7, 5) node[left]{\scriptsize dummy};

	  \draw (r13) edge (- 0.2 , 4);
          \draw (r13) edge (0  , 4);
          \draw (r13) edge (0.2 , 4);

          \draw (r33) edge (7.8 , 4);
          \draw (r33) edge (8 , 4);
          \draw (r33) edge (8.2 , 4);


     \end{tikzpicture}
 }\end{subfigure}

 \caption{Set Cover instance and constructed MFASP instance $\hat{G}$}
 \label{fig:SetCoverHardness}
 \end{figure}

We now analyze the structure of the instance we built:
 
\begin{restatable}{claim}{claimGED}
The Gallai-Edmonds decomposition of $G$ is $X\cup Y\cup Z$ where $Z=\emptyset$, $Y=Y'$ and $X=\hat{V}\setminus Y'$.
\end{restatable}
\begin{proof}
 Using the Tutte-Berge formula for the set
 $W = \{S_j^i: 1 \leq j \leq m, 1 \leq i \leq n  \}$,
 we see that a maximum matching has size at most $nm(N + 1) +  \sum_{i = 1}^n \lfloor F_i/2\rfloor$. 
 Clearly, a matching of this size exposing a vertex $v$ can be constructed for any $v \in C_j^i$ or $v\in Q_i$. 
 Moreover, a matching of this size cannot be constructed by exposing a vertex $v=S_j^i$: If so, such a vertex $S_j^i$ 
 would belong to a factor-critical component $K$ in the Gallai-Edmonds decomposition and $K$ also contains $C_j^i$. 
 But factor-critical graphs are 2-edge connected, and removing the edge $(S_j^i, c_j^i)$ would separate the graph $K$ into two components, a contradiction.
\end{proof}

\begin{restatable}{claim}{claimSetCoverMFASP}\label{claim:set-cover-to-mfasp}
Let $T\subseteq[m]$ denote the indices corresponding to a set cover.
Then there exists a feasible solution to the MFASP instance $\hat{G}$ whose stabilizer cost is at most $n(1+|T|/2)$.
\end{restatable}
\begin{proof}
For each $i\in [n]$, let $k_i$ denote an arbitrarily chosen index in $T$ such that the set $S_{k_i}$ contains the element $x_i$.

Consider a matching $\bar{M}$ obtained by matching $S_j^i$ with $c_j^i$ $\forall j\in[m],i\in[n]$,
picking a perfect matching of the rest of the clique vertices $V(C_j^i)\setminus \{c_j^i\}$,
exposing $x_i^{k_i}$ and picking a perfect matching of the rest of the vertices in each odd cycle $Q_i$.

Obtain a fractional vertex cover solution $\bar{y}$ as follows: For every $i\in [n]$,
let $\bar{y}_{S_j^i}=1$ if $j\in T$ and $\bar{y}_{S_j^i}=\nicefrac{1}{2}$ if $j\in [m]\setminus T$.
For each $i\in [n]$, set $\bar{y}_{x_{k_i}}=0$, $\bar{y}_{x_{k}}=1$
for the two vertices $x_k$ in $Q_i$ that are adjacent to $x_{k_i}$ and $\bar{y}_{x_{k}}=\nicefrac{1}{2}$ for the other vertices in $Q_i$.

Obtain the solution $\bar{c}$ as $\bar{c}_{uv}=\bar{y}_u+\bar{y}_v-1$ for every $uv\in \bar{M}$.
Then the solution $(\bar{M},\bar{y},\bar{c})$ is a feasible solution to the MFASP instance $\hat{G}$.
Moreover, the cost of the stabilizer $\1^T\bar{c}$ is $n(1+|T|/2)$.
\end{proof}

Let $(M,y,c)$ be an $f$-approximate feasible solution to the MFASP instance $\hat{G}$. We now can assume the following properties. 
If these are not fulfilled, we can change the solution without increasing the cost.   
\begin{restatable}{claim}{claimStructureSol}\label{claim:structure-of-mfasp-soln}
We can assume the following properties.
\begin{enumerate}
\item $M$ matches $S_j^i$ to $c_j^i$ for every $j\in [m],\ i\in [n]$ and $y_v=\nicefrac{1}{2}$ for every $v\in V(C_j^i)$.
\item $y_{S_j^i}=y_{S_j^1}$ for every $i\in [n]$, $j\in [m]$.
\end{enumerate}
\end{restatable}
\begin{proof}
We split the proof into two parts and show both properties separately.
\begin{enumerate}
\item If this was not the case, then there is at least one clique $C_j^i$ with a vertex $v$ with $y_v = 0$. 
Thus, $y_w = 1$ for all $w \in V(C_j^i) \setminus v$. 
By the complementary slackness condition given in Section \ref{sec:prelims}, we have $y_j+y_k=1+c_{jk}$ for each matching edge $\{j,k\}\in M$. 
Thus, $\sum_{e \in M \cap E(C_j^i)} c_e = |M \cap E(C_j^i)| = N$. However, we note that $T=[m]$ is a feasible set cover and 
by Claim \ref{claim:set-cover-to-mfasp}, this gives a feasible stabilizer of cost at most $nm/2+n$.
\item If $y_{S_j^i}\neq y_{S_j^1}$, then consider the block $i_0$ such that $\sum_{j\in [m]}y_{S_j^i}$ is minimum. 
Since the neighborhood of $\{S_1^{i_0},\ldots,S_m^{i_0}\}$ is identical to the neighborhood of all other blocks, 
we may copy the same vertex cover values $y$ for all other blocks and obtain a stabilizer with non-increasing cost.
\end{enumerate}
\end{proof}

Therefore, the set of $M$-exposed vertices contains exactly one vertex in each odd cycle $Q_i$.
Moreover, we can assume an exposed vertex is not a dummy vertex $x_i^{F_i + 1}$.
Otherwise, we could change that without increasing the cost of the stabilizer.

\begin{claim}\label{claim:mfasp-to-set-cover}
Let $X$ be the $M$-exposed vertices. Let $P:=\{j\in [m]:S_j^1\in N(X)\}$.
Then $\{S_j:j\in P\}$ is a set cover of cardinality at most $2f|P^*|$, where $P^*$ is the set of indices corresponding to the optimal set cover.
\end{claim}
\begin{proof}
We first show that $P$ is indeed a set cover. Consider an element $x_i$. By Claim \ref{claim:structure-of-mfasp-soln},
there exists $k\in [F_i]$ such that $x_i^k$ is $M$-exposed. There exists a set $S_r^1$ adjacent to $x_i^k$. Hence $r\in P$ and thus $S_r$ covers $x_i$.

It remains to bound the cardinality of $P$. Since $M$ exposes exactly one vertex in each odd cycle $Q_i$, by complementary slackness conditions,
$\sum_{e\in Q_i}c_e\ge 1$. Thus, $\sum_{i\in [n]}\sum_{e\in Q_i}c_e\ge n$. Let $r\in P$.
So $S_r^1$ is adjacent to an $M$-exposed vertex in $Q_i$ and hence $y_{S_r^1}\ge 1$.
By Claim \ref{claim:structure-of-mfasp-soln}, we have that $y_{S_r^i}\ge 1$ for every $i\in [n]$.
Since $y_v=\nicefrac{1}{2}$ for every $v\in V(C_j^i)$ (using Proposition \ref{lemma:nontrivialproperties}),
we have that $c_{S_r^i,c_r^i}\ge \nicefrac{1}{2}$. Thus, $\sum_{r\in P}\sum_{i\in [n]}c_{S_r^i,c_r^i}\ge |P|n/2$.
Thus, the cost of the stabilizer $\1^Tc\ge n(1+|P|/2)$. Hence, $|P|\le 2(\1^Tc/n-1)\le 2(f\1^Tc^*/n-1)$.

By Claim \ref{claim:set-cover-to-mfasp}, we have that the cost of the optimal stabilizer $\1^Tc^*$ is at most $n(1+|P^*|/2)$.
Thus, $|P|\le f|P^*|+2(f-1)\le 2f|P^*|$.
\end{proof}
Theorem \ref{thm:MFASP-SC-hardness} follows by Claim \ref{claim:mfasp-to-set-cover}.
\end{proof}

We now show that Theorem~\ref{thm:MFASP-SC-hardness} implies the second part of Theorem~\ref{thm:hardness}. Note that Dinur and Steurer \cite{DiStSetCoverHard} showed that there is no $(\log(n)-\epsilon)$-approximation algorithm for Set Cover, even if the number of sets $m$ is at most $n^2$, unless P=NP. Theorem \ref{thm:MFASP-SC-hardness} implies that there is no $\nicefrac{1}{2}(\log(n)-\epsilon)$-approximation for MFASP, where $n$ is the number of elements of the corresponding Set Cover instance. Now let $\hat{n}$ denote the number of vertices in the MFASP instance constructed in the proof of Theorem \ref{thm:MFASP-SC-hardness}. We have 
\[\hat{n} \leq nm\left(2(nm)^2+3\right)+nm \leq 6(nm)^3 \leq 6n^9.\]

Hence, for $\hat{n}$ being the number of vertices in an MFASP instace, we conclude that unless P=NP, there is no approximation algorithm for MFASP with approximation factor better than $(1/20)\log(\hat{n})-\epsilon$.

  

\section{An OPT-approximation in graphs with no singletons}\label{sec:sqrt-n-approx}
In this section, we present an algorithm that achieves a
$\min\{\sqrt{n}, OPT\}$-approxi\-ma\-ti\-on factor in graphs whose
Gallai-Edmonds decomposition has no trivial factor critical
components.  As a subroutine, we use an extension of an algorithm to
solve MFASP in factor-critical graphs. We mention how to do this by
solving an LP. Note that this is also possible using ``combinatorial techniques'' (in particular, without solving an LP) 
by computing a certain minimum vertex cover in a constructed bipartite graph. 
However, we will not go into details here. 

Our main theorem is the following:

\thmOPTApprox*

In the remainder of this Section, we prove Theorem
\ref{thm:opt-approx}. We describe the algorithm as part of the proof, but for an overview, we also give the pseudocode at the end of this section.   
Fix an optimum solution $(M^*, y^*, c^*)$
satisfying the properties (c) and (d) given in Theorem \ref{thm:opt-properties} and Lemma \ref{lemma:PositiveTutteSet}. 
Then, by Proposition \ref{lemma:nontrivialproperties}, 
$c^*_e = 0$ for $e$ in a component that is not $M^*$-exposed. Moreover, as mentioned before, w.l.o.g.\ $Z = \emptyset$. 
As usual, $OPT := \sum_{e \in E} c^*(e)$.  Let $r$ denote the
difference between the number of components in $G[X]$ and the number
of vertices in $Y$.  As $M^*$ is a maximum matching, the properties of
the Gallai-Edmonds decomposition imply that $M^*$ exposes exactly $r$
vertices, at most one in each component of $G[X]$.  Further, $M^*$
matches at most one vertex of a component to a vertex in $Y$, while
the rest are matched within the component.

For each factor-critical component $K$, we compute a lower bound on
the cost of an optimum stabilizer where the matching exposes $K$.
\begin{lemma}\label{lem:cost-lower-bound}
Let $K$ be a (non-trivial) factor-critical component in $G[X]$. For each vertex $w$ in $K$, let $\ell_{K,w}$ denote the optimum value of the following LP:
\begin{align*}
\ell_{K,w} :=\min ~ & \sum_{v\in V(K)\cup N_G(V(K))} y_v - \left(\frac{|V(K)|-1}{2}\right) - \frac{|N_G(V(K))|}{2}\\
~~~~ \st\ & y_i + y_j \ge 1\ \forall\ \{i,j\}\in E[V(K)\cup N_G(V(K))]\\
& y_i\ge 1/2\ \forall\ i\in N_G(V(K))\\
& y_i\ge 0\ \forall\ i\in V(K)\\
& y_w=0
\end{align*}
Let $f(K):=\min_{w\in V(K)}\ell_{K,w}$. If $M^*$ exposes
a vertex in $K$ then $\1^Tc^* \geq f(K)+r-1$. 
\end{lemma}

We now show that Lemma \ref{lem:cost-lower-bound} can be used to obtain an optimal solution for MFASP in factor-critical graphs and thereby prove Theorem \ref{theorem:FactorCriticalGraphs}. 
We note that factor-critical graphs are the special case where $G$ consists of one component $K$ (and thus $N(V(K)) = \emptyset$). In that case an optimum stabilizer can be obtained by computing $f(K)$, choosing any matching $M^*$ exposing $w^*=\text{argmin}_{w \in V(K)}\ell_{K, w}$ and setting $c^*$ to fulfill complementary slackness (i.e., if $y^*$ is a solution for $\ell_{K,w^*}$, then set $c^*(uv):=y^*_u+y^*_v-1$ for every $uv\in M^*$ and $c^*(uv)=0$ for every $uv\in E\setminus M$). 

\begin{proof}[Proof of Lemma \ref{lem:cost-lower-bound}]
Recall that $(M^*,y^*,c^*)$ is an optimum MFASP solution that 
satisfies the properties (c) and (d) of Theorem
\ref{thm:opt-properties}. We then have
\begin{align*}
\1^Tc^* = \sum_{e\in M^*} c^*_{e} 
= \!\!\!\!\!  \sum_{K'\neq K: K'\text{ is }M^*\text{-exposed}} \sum_{e\in M^*\cap E(K')} \!\!\!c^*_{e} + \!\!\!\!\! 
\sum_{e\in M^*\cap E(K)}\!\!\! c^*_e + \!\!\!\!\! 
\sum_{e\in M^*\cap \delta_G(Y)} \!\!\!c^*_e.
\end{align*}

The first double-sum on the right-hand side is at least $r-1$ by Lemma
\ref{lemma:FactorCriticalNotForFree}. (This Lemma only can be applied, because the factor-critical components are non-trivial.)  
By complementary slackness
conditions as mentioned in Section \ref{sec:prelims}, we know that for
every edge $\{i,j\}\in M^*$, we have $1+c^*_{ij}=y_i^*+y_j^*$. As $M^*$
exposes one vertex in K,
\[\sum_{e\in M^*\cap E(K)} c^*_e = \sum_{\{i,j\}\in M^*\cap E(K)} (y_i^* + y_j^* -1)  =  \sum_{v\in V(K)} y_v^* - \left(\frac{|V(K)|-1}{2}\right).\]
If $\{i,j\}\in M^*$ with $i\in Y$, then $j\in X$ is a vertex in a
factor-critical component that is matched by $M^*$.  By Proposition
\ref{lemma:nontrivialproperties}, we have that $y_j^*=1/2$.  Hence,
\begin{align*}
\sum_{e\in M^*\cap \delta_G(Y)} c^*_e &= \sum_{\{i,j\}\in M^*: i\in Y, j\in X} (y_i^* + y_j^* -1) \\
&= \sum_{i\in Y} \left(y_i^* -\frac{1}{2}\right)
\ge \smashoperator[r]{\sum_{i\in N_G(V(K))}}y_i^* -\frac{|N_G(V(K))|}{2}.
\end{align*}
Let $w$ be a $M^*$-exposed vertex in $K$. Then, $y^*$ restricted to the vertices $V(K)\cup N_G(V(K))$ 
is a feasible solution to the LP corresponding to $\ell_{K,w}$. Combining the three relations, we get that
\begin{align*}
\sum_{e\in M^*} c^*_{e} &\ge r-1+\hspace{-4mm}\sum_{v\in V(K)\cup N_G(V(K))}\hspace{-2mm} y_v^* - \left(\frac{|V(K)|-1}{2}\right) - \frac{|N_G(V(K))|}{2} \\
&\ge r-1+\ell_{K,w}\ge r-1+f(K).
\end{align*}
\end{proof}

In order to identify a suitable matching to stabilize, 
we build an auxiliary graph $G'$ as follows: 
Contract each component $K$ in $G[X]$ to a pseudo-vertex $v_K$ 
and assign edge weight $w_e:=f(K)$ for all edges $e$ incident to the contracted vertex $v_K$. 
Compute a matching $M$ in $G'$ of maximum weight covering $Y$. 

\begin{lemma}\label{lem:OPT-lower-bound}
The cost $\1^Tc^*$ of an optimum stabilizer $(M^*,c^*,y^*)$ is at least 
\[
r-1+\max_{K:v_K\text{ is }M\text{-exposed}}f(K).
\]
\end{lemma}
\begin{proof}
  Let $K=\argmax_{K:v_K\text{ is }M\text{-exposed}}f(K)$.  If $M^*$
  exposes $K$, then Lemma \ref{lem:cost-lower-bound} proves the claim. So,
  we may assume that $M^*$ matches $K$.  Consider $M^*$ restricted to
  the edges in the bipartite graph $G'$. Both $M$ and $M^*$ are
  maximum cardinality matchings in $G'$ and $v_K$ is $M$-exposed.  So,
  we have an $M$-alternating path $P$ starting from $v_K$ and ending
  at another vertex corresponding to a contracted factor-critical
  component. Let
  $P=v_{K_1},b_1,v_{K_2},b_2,\ldots,v_{K_{t-1}},b_{t-1},v_{K_t}$ for
  some $t\ge 1$ and $v_{K_1}=v_K$ and $v_{K_t}$ is $M^*$-exposed.
  Since $M$ is a maximum weight matching, we have
  $\sum_{e\in M\Delta P} w_e\le \sum_{e\in M\cap P}w_e$.  Thus,
  $\sum_{i=1}^{t-1}f(K_i) \le \sum_{i=2}^{t}f(K_i)$ and we have that
  $f(K)=f(K_1)\le f(K_t)$. Thus, the cost of the stabilizer $c^*$ is
  at least $r-1+f(K_t)\ge r-1+f(K)$.
\end{proof}

We now stabilize $M$. For each $M$-exposed vertex $v_K$,
let \[w_K:=\text{argmin}_{w\in V(K)}\ell_{K,w},\] and let
$\overline{y}^{w_K}$ denote the solution $y$ achieving the optimum for
$\ell_{K,w_K}$.  Extend $M$ inside each factor-critical component $K$:
if $v_K$ is matched by $M$ using edge $\{u,b\}$ where
$u\in V(K), b\in Y$, then extend $M$ using a matching in $K$ that
exposes $u$. If $v_K$ is exposed by $M$, extend $M$ using a matching
in $K$ that exposes $w_K$. Let $\overline{M}$ denote the resulting
matching.

For each vertex $v_K$ matched by $M$, set $\overline{y}_u=1/2$ for all vertices $u\in V(K)$. 
For each vertex $v_K$ that is exposed by $M$, set $\overline{y}_u=\overline{y}^{w_K}_u$ for all vertices $u\in V(K)$. 
For each vertex $b\in Y$ that is adjacent to a $M$-exposed $v_K$, set $\overline{y}_b=\max_{K:v_K\text{ is }M\text{-exposed}} \overline{y}^{w_K}_b$. For each vertex $b\in Y$ with no adjacent $M$-exposed $v_K$, set $\overline{y}_b=1/2$. 
Note that these are only good choices because no trivial factor-critical components exist. 
For trivial components, there are cases where (for any reasonably good solution) even though the trivial component is matched, its $y$-value must be $0$.

Set $\overline{c}(uv)=\overline{y}_u+\overline{y}_v-1$ for edges $\{u,v\}\in \overline{M}$ and $\overline{c}(uv)=0$
for edges $\{u,v\}\in E\setminus \overline{M}$. 

We next show that the solution $(\overline{M},\overline{y},\overline{c})$ is a feasible solution. 
\begin{lemma}\label{lem:feasible-soln}
$(\overline{M},\overline{y},\overline{c})$ is a feasible solution to MFASP.
\end{lemma}
\begin{proof}
By construction, $\overline{M}$ is a matching and $\sum_{e\in \overline{M}}(1+\overline{c}_e)=\sum_{e\in \overline{M}}(\overline{y}_u+\overline{y}_v)$. It remains to show that $\overline{y}$ is a feasible fractional $w$-vertex cover for $w_e=1+\overline{c}_e$ for every $e\in E$. 

Consider an edge $e=\{u,v\}\in E$. If $e\in \overline{M}$, then $\overline{y}_u+\overline{y}_v=1+\overline{c}_{uv}$. Let $e\in E\setminus \overline{M}$. For such edges, we have $\overline{c}_e=0$ and hence $1+\overline{c}_e=1$. 

We distinguish several cases. 
If $e\in K$ where $v_K$ is matched by $M$, then $\overline{y}_u=\overline{y}_v=1/2$ and hence $\overline{y}_u+\overline{y}_v=1$. 
If $e\in K$ where $v_K$ is exposed by $M$, then $\overline{y}_u+\overline{y}_v=\overline{y}^{w_K}_u+\overline{y}^{w_K}_v\ge 1$ by the feasibility of the solution $\overline{y}^{w_K}$ to the LP corresponding to $\ell_{K,w}$. 
If $e\in \delta_G(Y)$, then let $u\in Y, v\in V(K)$. 
If $v\in V(K)$ where $v_K$ is matched by $M$, then $\overline{y}_v=1/2$ and moreover $\overline{y}_u\ge 1/2$ and hence $\overline{y}_u+\overline{y}_v\ge 1$. 
If $v\in V(K)$ where $v_K$ is exposed by $M$, then $\overline{y}_v=\overline{y}^{w_K}_v$ and $\overline{y}_u=\max_{w_K:v_K\text{ is }M\text{-exposed}}\overline{y}^{w_K}_u\ge \overline{y}^{w_K}_u$. By the feasibility of the solution $\overline{y}^{w_K}$ to the LP corresponding to $\ell_{K,w}$, we have that $\overline{y}_u+\overline{y}_v\ge \overline{y}^{w_K}_u+\overline{y}^{w_K}_v\ge 1$.
\end{proof}

We now bound the cost of the constructed solution $(\overline{M},\overline{y},\overline{c})$. 

\begin{lemma}\label{lem:cost-of-soln}
The cost $\1^T \overline{c}$ of the stabilizer $(\overline{M},\overline{y},\overline{c})$ is at most $(\sum_{e\in E}c^*_e)^2$.
\end{lemma}
\begin{proof}
Let $\mathcal{K}$ be the set of components such that $v_K\text{ is }M\text{-exposed}$. The cost of $(\overline{M},\overline{y},\overline{c})$ is
\[\sum_{e\in \overline{M}}\overline{c}_e = \sum_{K \in \mathcal{K}}\sum_{\{u,v\}\in \overline{M}\cap K}(\overline{y}_u+\overline{y}_v-1) 
+ \sum_{u\in Y, v\in X: \{u,v\}\in M}(\overline{y}_u+\overline{y}_v-1)\]


We next bound the second term in the above sum using $y_v = \nicefrac{1}{2}$ for $v \in X$ with $\{u, v\} \in M$. 
Let \[Y' = \{u \in Y: u \text{ is not adjacent to an $M$-exposed vertex}\}.\] For $u \in Y\setminus Y'$, we have $\overline{y}_u=1/2$ 
and such vertices do not contribute to the sum.
%
\begin{align*}
\sum_{u\in Y}\left(\overline{y}_u-\frac{1}{2}\right) &\le
  \sum_{u\in Y'} \left(\sum_{K \in \mathcal{K}: u \in N_G(V(K))}\left(\overline{y}^{w_K}_u-\frac{1}{2}\right)\right)\\
  &= \sum_{K\in \mathcal{K}}\left(\sum_{u\in Y\cap N_G(V(K))} \overline{y}^{w_K}_u-\frac{|N_G(V(K))|}{2}\right)
  \end{align*}


Therefore,
\begin{align*}
 &\sum_{e\in \bar{M}}\overline{c}_e \le  
\sum_{K \in \mathcal{K} }\left(\sum_{v\in V(K)\cup N_G(V(K))}\overline{y}^{w_K}_v-\frac{|V(K)|-1}{2}-\frac{|N_G(V(K))|}{2}\right)\\
&\le  r\max_{K \in \mathcal{K}} f(K)\\
&{\le} \left(\frac{r+\max_{K \in \mathcal{K}} f(K)}{2}\right)^2 \quad \quad \text{(since arithmetic mean is at least goemetric mean)}\\
&{\le} \left(\frac{1+\sum_{e\in E}c^*_e}{2}\right)^2 \quad \quad \quad \quad \text{(by Lemma \ref{lem:OPT-lower-bound})}\\
&\le \left(\sum_{e\in E}c^*_e\right)^2.
\end{align*}

\end{proof}

If $OPT > \sqrt{n}$, any solution fulfilling properties (a) (b), (c) and (d) 
of Theorem \ref{thm:opt-properties} is a $\sqrt{n}$-approximation as
the cost of any such solution is bounded by
$\nu(G) \le \nicefrac{n}{2}$. Therefore, Lemmas \ref{lem:cost-of-soln} and \ref{lem:feasible-soln} and the construction of $(\overline{M},\overline{y},\overline{c})$ 
imply Theorem \ref{thm:opt-approx}. We give an overview of the algorithm here:

\medskip
\noindent
\noindent\hrule
\vspace{0.1cm}
\textbf{Algorithm}
\vspace{0.1cm}
\noindent\hrule
\begin{enumerate}
\item For each factor-critical component $K$ in $G[X]$:
\begin{enumerate}
\item For each vertex $w$ in $K$, solve following LP:
\begin{align*}
\ell_{K,w} :=\min \sum_{v\in V(K)\cup N_G(V(K))} y_v &- \left(\frac{|V(K)|-1}{2}\right) - \frac{|N_G(V(K))|}{2}\\
y_i + y_j &\ge 1\ \forall\ \{i,j\}\in E[V(K)\cup N_G(V(K))]\\
y_i&\ge 1/2\ \forall\ i\in N_G(V(K))\\
y_i&\ge 0\ \forall\ i\in V(K)\\
y_w&=0
\end{align*}
\item Let $f(K):=\min_{w\in V(K)}\ell_{K,w}$. 
\end{enumerate} 

\item Construct an auxiliary bipartite graph $G'$ from $G$ as follows: 
Contract each component $K$ in $G[X]$ to a pseudo-vertex $v_K$ and assign edge weight $w_e:=f(K)$ for all edges $e$ incident to the contracted vertex $v_K$. 
Delete edges in $E[Y]$.

\item Compute a matching $M$ in $G'$ of maximum weight covering $Y$. 

\item For each $M$-exposed vertex $v_K$, let $w_K:=\text{argmin}_{w\in V(K)}\ell_{K,w}$, 
let $\overline{y}^{w_K}$ denote the solution $y$ achieving the optimum for $\ell_{K,w_K}$. 

\item Identify a matching $\overline{M}$: Extend $M$ inside each factor-critical component $K$: 
if $v_K$ is matched by $M$ using edge ${u,b}$ where $u\in V(K), b\in Y$, then extend $M$ using a matching in $K$ that exposes $u$. 
If $v_K$ is exposed by $M$, extend $M$ using a matching in $K$ that exposes $w_K$. Let $\overline{M}$ denote the resulting matching. 

\item Identify a fractional vertex cover $\overline{y}$: For each vertex $v_K$ matched by $M$, set $\overline{y}_u=1/2$ for all vertices $u\in V(K)$. 
For each vertex $v_K$ that is exposed by $M$, set $\overline{y}_u=\overline{y}^{w_K}_u$ for all vertices $u\in V(K)$. 
For each vertex $b\in Y$ that is adjacent to a $M$-exposed $v_K$, set $\overline{y}_b=\max_{K:v_K\text{ is }M\text{-exposed}} \overline{y}^{w_K}_b$. 
For each vertex $b\in Y$ with no adjacent $M$-exposed $v_K$, set $\overline{y}_b=1/2$. 

\item Identify a feasible MFASP solution $\overline{c}$: 
Set $\overline{c}(uv)=\overline{y}_u+\overline{y}_v-1$ for edges $\{u,v\}\in \overline{M}$ and $\overline{c}(uv)=0$ for edges $\{u,v\}\in E\setminus \overline{M}$. 

\item Return $(\overline{M},\overline{y},\overline{c})$.
\end{enumerate}
\noindent\hrule


\section{An Exact Algorithm for MFASP}\label{sec:ExactAlgo}

In this section, we describe an exact algorithm to solve the MFASP in
arbitrary graphs $G$.  The algorithm is based on the Gallai-Edmonds
decomposition $V(G)=X\cup Y \cup Z$ and makes use of a polynomial time exact
algorithm to solve MFASP in the factor-critical components in $G[X]$.
The runtime of our algorithm grows
exponentially only in the size of the Tutte set.  Thus, our algorithm
is fixed parameter tractable w.r.t.\ the size of the Tutte set $Y$.
In particular, the resulting algorithm runs in polynomial-time if the
size of the Tutte set is bounded by $\mathcal{O}(\log n)$.

\thmExactAlgo*
  
\noindent \textbf{Outline of the algorithm.} 
By property $(ii)$ in Theorem \ref{thm:opt-properties}, 
we know that there exists a subset $S^*\subseteq Y$ and a half-integral minimum fractional stabilizer $c^*$ 
with a half-integral minimum fractional $(1+c^*)$-vertex cover solution $y^*$ such that $y^*_v=1$ for all $v\in S^*$, 
and $y_v^*=\frac{1}{2}$ for all $v\in Y\setminus{S^*}$. 
This motivates the following problem: given a set $\hat{S}\subseteq Y$, find a minimum fractional stabilizer $c$ which admits a minimum fractional $(1+c)$-vertex cover $y$
satisfying $y_v = 1$ if $v\in \hat{S}$ and $y_v=\frac{1}{2}$ if $v\in Y\setminus{\hat{S}}$.
Or, decide that no such solution exists. 
In Section \ref{sec:fixedTutteAlgo} we present a polynomial-time algorithm for this problem. 
Repeatedly applying this algorithm to all subsets of the Tutte set and searching for the optimal one gives the optimal solution to 
MFASP and implies Theorem \ref{theorem:exact-algorithm-for-MFASP}.

\subsection{Algorithm to find the optimal stabilizer knowing the subset of Tutte vertices with y-value one}\label{sec:fixedTutteAlgo}
Let $G$ be a graph with Gallai-Edmonds decomposition $X, Y, Z$. In this section, we focus on the following problem: 
Given a set $\hat{S}\subseteq Y$, find a minimum cost fractional additive stabilizer $(M,y,c)$ 
among those which have $y_v=1$ if $v\in \hat{S}$, and $y_v=\frac{1}{2}$ if $v\in Y\setminus{\hat{S}}$. 
Let $f(\hat{S})=\sum_{e\in E} c_e$ denote the cost of such a solution. 

We give an overview of the algorithm to compute $f(\hat{S})$. 
(For a formal description see Algorithm $MFASP(\hat{S})$.) 
Let $(M,y,c)$ denote the triple of an optimal solution corresponding to $f(\hat{S})$. 
Let us examine the structure of the optimal solution $(M,y,c)$. Recall that we are restricting $c$ and $y$ to be half-integral. \\



\noindent \textbf{Finding an optimum with knowledge of matching edges between $Y$ and $X$.}
Let us focus on the matching edges in $M$ that link $Y$ to $X$ and argue that it is sufficient to know these links to find an optimal solution $c$. 
We consider a component $K\in G[X]$ matched to some vertex in $v \in Y$ by $M$ and distinguish two cases.

(i) $K$ is non-trivial. By Proposition \ref{lemma:nontrivialproperties}, we have $c(e) = 0$ for $e \in E[K]$ and $c(\delta(V(K)) = y_v - \nicefrac{1}{2}$.


(ii) Suppose $K = \{u\}$. If $u$ is not incident to any vertex $w \in Y\setminus \hat{S}$, 
then we may assign $y_u=0$ thereby incurring a cost of $c_e=0$ and this is optimal. 
Otherwise, the feasible $y$ assigns $y_u=1/2$
and as before incurs an optimal cost of $y_v - \nicefrac{1}{2}$ over $\delta(V(K))$. 

Next let us consider $K\in G[X]$ that is not matched to any vertex in $Y$. 

(i) Suppose there are no edges $\{v,u\}\in \delta(V(K))$ that are incident to a vertex $u\in Y\setminus \hat{S}$. 
Then the optimal fractional additive stabilizer $c$ restricted to the set of edges in $E(K)\cup \delta(V(K))$ 
should also be an optimal fractional additive stabilizer for $K$ itself and vice-versa. 
Therefore, the stabilizer values on these edges can be computed using the exact algorithm for factor-critical graphs.

(ii) Suppose there are edges $\{v,u\} \in \delta(V(K))$ that are incident to a vertex $u \in Y\setminus \hat{S}$.
In this case, the vertex cover values $y$ should satisfy the covering constraints for the edges in $\delta(V(K))$.
In particular, $y_v\ge 1/2$ for vertices $v\in V(K)$ which have neighbors in $Y\setminus \hat{S}$. 
As a consequence, the optimal stabilizer restricted to the set of edges in $E(K)\cup \delta(V(K))$ may not be the optimal stabilizer for $K$ itself. However, the optimal fractional additive stabilizer restricted to the set of edges in $E(K)\cup \delta(V(K))$ should also be an optimal fractional additive stabilizer for a modified graph $\tilde{K}$ obtained from $K$ by adding an extra loop $\{v,v\}$ to each vertex $v\in V(K)$ linked to a vertex $u\in Y\setminus{\hat{S}}$. Conversely, we can modify an optimal fractional additive stabilizer $c$ over the set of edges in $E(K)\cup \delta(V(K))$ to take the same values as an optimal fractional additive stabilizer for $\tilde{K}$ without losing optimality. Further, we note that we can compute a minimum fractional additive stabilizer in $\tilde{K}$, by running the algorithm $A(v)$ for every node $v$ that is not incident to a loop in $\tilde{K}$ and output the best.

We observe that if every vertex $v\in V(K)$ has an edge adjacent to a vertex $u\in Y\setminus \hat{S}$, then this necessiates $y_v\ge 1/2$ for every vertex in $K$ and therefore $K$ must necessarily be matched to a vertex in $Y$. \\


\noindent \textbf{Computing matching edges between $Y$ and $X$.}
From the above discussion, it is clear that the cost of the solution $f(\hat{S})$ does not depend on the precise choice of the edges used to match the components of $G[X]$ to $Y$ but only depends on which components of $G[X]$ are matched by $M$. Therefore, 
we can also identify the edges between $Y$ and $X$ in an optimal matching $M$ as follows: 
Let us denote by $\kappa(K, \hat{S})$ the cost of the stabilizer over the edges in $E[K]\cup \delta(V[K])$ if $K$ is not matched to $Y$ 
(as observed before, we can compute $\kappa(K, \hat{S})$ by applying an exact algorithm to the factor-critical graph $\tilde{K}$, 
for example the one presented in section \ref{sec:sqrt-n-approx}).
If $K$ must necessarily be matched, then we set $\kappa(K, \hat{S})$ to infinity
(or an arbitrarily large value $U$ in the implementation). Let $T$ denote the trivial components in $G[X]$ all of whose neighbors are in $\hat{S}$. 
Let us construct a weighted bipartite graph $H$ from $G$ as follows: Delete $Z$, delete the edges between vertices in $Y$, and contract each component $K$ of $G[X]$ to a vertex $v_K$; replace the multi-edges by a single edge to make it a simple graph and 
for a vertex $u\in Y$ that is adjacent to some node in $K$, we introduce weight $\kappa(K,\hat{S})$ on the edge $\{u,v_K\}$. 

By the above discussion, a maximum weight matching $N$ in $H$ covering all vertices in $Y$ gives the edges of an optimal matching $M$ between $Y$ and $X$. Therefore, 
\[
f(\hat{S})=\frac{1}{2}\left(|\hat{S}|-|\{K\in T: N \text{ covers } v_K \}|\right)+\sum_{K\in G[X]:\ v_K \text{ is exposed by }N}\kappa(K,\hat{S}).
 \] 
Hence, we find a maximum weight matching in $H$ to compute $f(\hat{S})$.

\begin{remark}
Between matching a component in $T$ or a non-trivial factor-critical component, 
$M$ prefers the latter choice by Lemma \ref{lemma:FactorCriticalNotForFree}. 
Thus, assigning $ \kappa(K, \hat{S}) = 0$ for $K \in T$, 
implicitly assumes that components in $T$ are only matched if there is no other choice.

 
\end{remark}

\medskip
 \noindent
\noindent\hrule
\vspace{0.1cm}
\noindent \textbf{Algorithm $MFASP(\hat{S})$.}
\vspace{0.1cm}
\noindent\hrule
 \hrule
 \begin{enumerate}
 \item For each factor-critical component $K$ in $G[X]$ compute the cost $\kappa(K, \hat{S})$ needed to stabilize $K$ 
 and the edges linking $K$ to vertices in $Y\setminus{\hat{S}}$ in case $K$ would not be matched to $Y$. 
 (We discussed above 
 that this can be done in polynomial time.) Let $(M^K, c^K, y^K)$ be an optimal stabilizer for $K\cup \delta(V(K))$ among those with $M^K\cap \delta(K)=\emptyset$. 
 \item Shrink the components $K$ in $G[X]$ to pseudo-vertices $v_K$, assign the weight $\kappa(K,\hat{S})$ 
 to all edges linking a Tutte vertex to pseudo-vertex $v_K$, and compute a bipartite matching $\hat{M}$ of maximum weight covering $Y$ (this is possible in polynomial time).
 
 If no feasible solution exists, i.e.\ if there exists an unmatched component $K$ with $\kappa(K, \hat{S}) = U$, STOP and RETURN INFEASIBLE.
 \item Obtain a maximum matching in $G$ by extending $\hat{M}$ as follows:
 \begin{itemize}
 \item for each component $K$ not matched to $Y$ add the matching edges in $M^K$  to $\hat{M}$;
 \item for each component $K$ having $v_K$ matched to $Y$, pick a vertex in $K$ that has the corresponding matching edge adjacent to it, say $v$, and add a maximum matching in $K$ that exposes $v$ to $\hat{M}$;
 \item for each  component $C$ in $G[Z]$ (we note that all these components are even), add an arbitrary perfect matching to $\hat{M}$;
 \end{itemize}
 \item Obtain a fractional additive stabilizer as follows:
 \begin{itemize}
 \item $\hat{c}_e=c_e^K$ for all components $K$ in $G[X]$ that are not matched to $Y$,
 \item $\hat{c}_e=\frac{1}{2}$ for each matching edge $e\in \hat{M}$ linked to a Tutte vertex $v\in \hat{S}$, 
      except if $e = \{v, w\}$ for some vertex $w$ that is a trivial component in $G[X]$ with $N_G(w)\subseteq \hat{S}$, and
 \item $\hat{c}_e=0$ else.
 \end{itemize}
 \item Obtain a fractional $(1+\hat{c})$-vertex cover $\hat{y}$ that satisfies complementary slackness with $\hat{M}$ as follows:
 \begin{itemize}
 \item $\hat{y}_v= \frac{1}{2}$ for all vertices in $Z$, all vertices in $Y\setminus{\hat{S}}$ 
 and all vertices in components $K$ in $G[X]$ that are matched to $Y$ except if $v$ is a trivial 
 component in $G[X]$ with $N_G(v)\subseteq \hat{S}$,
 \item $\hat{y}_v=1$ for all $v\in \hat{S}$, 
 \item $\hat{y}_v=y_v^K$ for all vertices in components $K$ in $G[X]$ that are not matched to $Y$, and
 \item $\hat{y}_v = 0$ for all vertices $v$ that are trivial components in $G[X]$ with $N_G(v)\subseteq \hat{S}$ and matched to $Y$.
 \end{itemize}
 \item Return $(\hat{M}, \hat{y}, \hat{c})$ and $f(\hat{S}):=\sum_{e\in E} \hat{c}_e$;
 \end{enumerate}
\noindent\hrule
\begin{remark}
As mentioned earlier, not every possible choice of $\hat{S}$ has a fractional additive stabilizer $c$ which has $y_v=1$ if $v\in \hat{S}$ and $y_v=1/2$ if $v\in Y\setminus \hat{S}$ for a half-integral minimum fractional $(1+c)$-vertex cover $y$. For example, consider a graph where $Z=\emptyset$, $Y = \{v\}$ and  $G[X]$ consists of two triangles whose nodes are all connected to $v$. Then $y_v$ must be $1$ and $\hat{S} = \emptyset$ is not feasible. The algorithm detects these cases in Step 2. 
\end{remark}

\subsection{An Approximation Algorithm for Graphs with Many Nontrivial Components}
We can also use the algorithm to compute $f(\hat{S})$ 
to obtain an approximation algorithm for graphs that have a 
large number of non-trivial factor-critical components in the Gallai-Edmonds decomposition.

\begin{restatable}{theorem}{thmApproxAlgo}
\label{thm_approxAlgo}
For a graph $G$ with Gallai-Edmonds decomposition
$V(G) = X \cup Y \cup Z$, let $\mathcal{C}^+$ denote the set of
nontrivial components in $X$. If
$|\mathcal{C}^+| \geq (1+\nicefrac{1}{k})|Y|$ for $k>0$, then there is
a $(\nicefrac{k}{2}+1)$-approximation algorithm for MFASP.
\end{restatable} 

\begin{proof}
Let $(M^*, y^*, c^*)$ be an optimal solution for MFASP with cost $|c^*|$ and $X_1, \ldots X_r$ be the components of $G[X]$. 
We first note that the number of unmatched non-trivial components of $X$ is at least $\ceil{\left(\nicefrac{1}{k}\right)|Y|}$. 
We know by Lemma \ref{lemma:FactorCriticalNotForFree} that the optimal stabilizer pays at least $1$ over the edges in each of these components. 
This yields
\begin{align*}
c^*(E) = \sum\limits_{e \in E}{c^*_e} \geq \sum\limits_{i=1}^r{\sum_{e \in X_i}{c^*_e}} \geq 
{ \frac{|Y|}{k} }.
\end{align*}

Let $\hat{S}=Y$, i.e., fix $y_t=1$ for all vertices $t$ in the Tutte set, and calculate an optimal solution $(\hat{M}, \hat{y}, \hat{c})$ 
corresponding to $f(\hat{S})$.
We observe that the optimal solution corresponding to $f(\hat{S})$ can be computed efficiently using the algorithm from Section \ref{sec:fixedTutteAlgo}.
Recall that $y^*_v\ge \nicefrac{1}{2}$ for each vertex $v$ in the Tutte set. Therefore, 
\begin{align*}
\hat{c}(E) &= \hat{y}(V) - |\hat{M}| = \sum\limits_{v \in V}{\hat{y}_v} - |M^*|
= \sum\limits_{v \in X}{\hat{y}_v} + \sum\limits_{v \in Y}{\hat{y}_v} - |M^*\setminus E[Z]|\\
&\leq \sum\limits_{v \in X}{y^*_v} + \frac{|Y|}{2} + \sum\limits_{v \in Y}{y^*_v} - |M^*\setminus E[Z]| \\
&\leq \sum\limits_{v \in X}{y^*_v} + \left(\frac{k}{2}\right)c^*(E) + \sum\limits_{v \in Y}{y^*_v} - |M^*\setminus E[Z]| 
= \left( \frac{k}{2}+1\right)c^*(E),
\end{align*}
which finishes the proof. In the first inequality above, we have used $\sum_{v \in X}{\hat{y}_v}\le \sum_{v \in X}{y^*_v}$ 
since $y^*$ restricted to $X$ is also feasible for the auxiliary problem with $y_v$ fixed to one on all Tutte vertices. 
\end{proof}

\bibliography{fractional_stabilizers}

\clearpage


\end{document}